\documentclass{IEEEtran}

\ifCLASSINFOpdf
   \usepackage[pdftex]{graphicx}
   \graphicspath{{../figures/}{../plots/}}
\else
   \usepackage[dvips]{graphicx}
   \graphicspath{{../figures/}{../plots/}}
\fi

\usepackage[cmex10]{amsmath}
\usepackage{amssymb,verbatim}
\usepackage{amsthm,bm}
\usepackage{algorithmic}
\usepackage{algorithm}
\usepackage[tight,footnotesize]{subfigure}

\newcommand{\B}{\mathcal{B}}
\newcommand{\C}{\mathcal{C}}
\newcommand{\F}{\mathcal{F}}
\newcommand{\myprob}[1]{\mathbb{P}\left\{#1\right\}}
\newcommand{\myexp}[1]{\mathbb{E}\left[#1\right]}
\newtheorem{lemma}{Lemma}

\newcommand{\nn}{\nonumber\\}

\usepackage[usenames,dvipsnames]{color}

\begin{document}

\title{Proportional Fair Coding for Wireless Mesh Networks}
\author{K.~Premkumar, Xiaomin Chen, and Douglas J.~Leith
\thanks{The authors are with the Hamilton Institute, 
        National University of Ireland Maynooth, Ireland.
		E--mail: \{Premkumar.Karumbu, Xiaomin.Chen,
		Doug.Leith\}@nuim.ie}
		

       \thanks{This work is supported by Science Foundation Ireland 
	   under Grant No. 07/IN.1/I901.}
       }
\maketitle

\begin{abstract}
We consider multi--hop wireless networks carrying unicast flows for
multiple users. Each flow has a specified delay deadline, and the lossy
wireless links are modelled as binary symmetric channels (BSCs). Since
transmission time, also called airtime, on the links is shared amongst 
flows, increasing the airtime for one flow comes at the cost of reducing 
the airtime available to other flows sharing the same link. We derive 
the joint allocation of flow airtimes and coding rates that achieves the
proportionally fair throughput allocation. This utility optimisation
problem is non--convex, and one of the technical contributions of this
paper is to show that the proportional fair  utility optimisation can
nevertheless be decomposed into a sequence of convex optimisation
problems. The solution to this sequence of convex problems is the unique
solution to the original non--convex optimisation. Surprisingly, this
solution can be written in an explicit form that yields considerable
insight into the nature of the proportional fair joint airtime/coding
rate allocation. To our knowledge, this is the first time that the
utility fair joint allocation of airtime/coding rate has been analysed,
and also, one of the first times that utility fairness with delay
deadlines has been considered.
\end{abstract}


\IEEEpeerreviewmaketitle

\section{Introduction}
In this paper, we consider wireless mesh networks with lossy links and
flow delay deadlines. Packets which are decoded after a delay deadline
are treated as losses. We derive the joint allocation of flow airtimes
and coding rates that achieves the proportionally fair throughput
allocation. To our knowledge, this is the first time that the utility
fair joint allocation of airtime/coding rate has been analysed, and
also, one of the first times that utility fairness with delay deadlines
has been considered (also, see \cite{li-atilla}, \cite{rsrikant}).

In the special cases where all links in a network are loss--free or all
flow delay deadlines are infinite, we show that the proportionally fair
utility optimisation decomposes into decoupled airtime and coding rate
allocation tasks. That is, a layered approach that separates MAC
scheduling and packet coding rate selection is optimal. This corresponds
to the current practice, and these tasks can be solved separately using a
wealth of classical techniques.

\begin{figure}
\centering
\subfigure[Topolopgy]
{
\includegraphics[width=0.35\columnwidth]{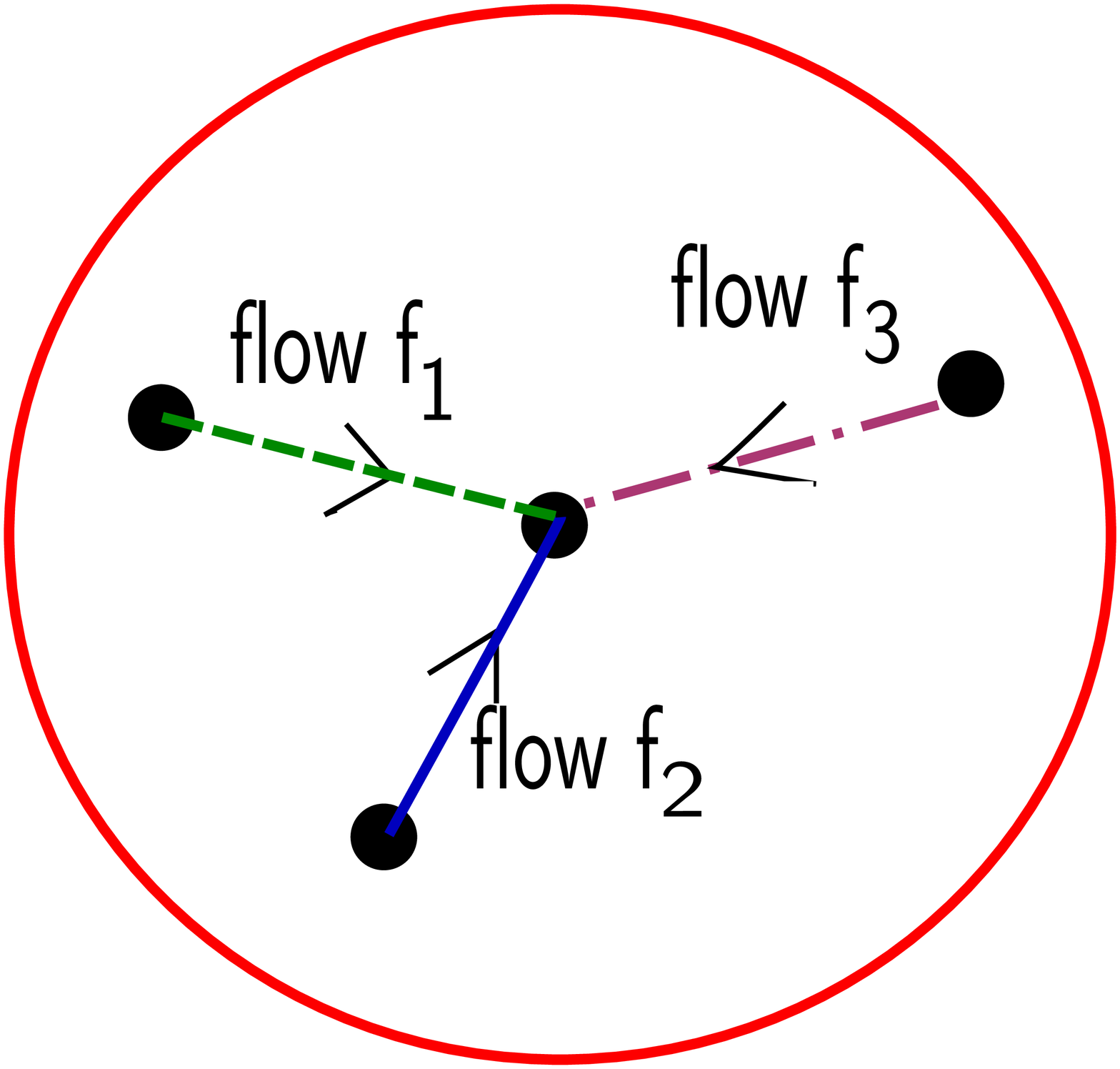}}
\subfigure[Optimum Utility vs Classical Utility]{
\includegraphics[width=50mm, height=40mm]{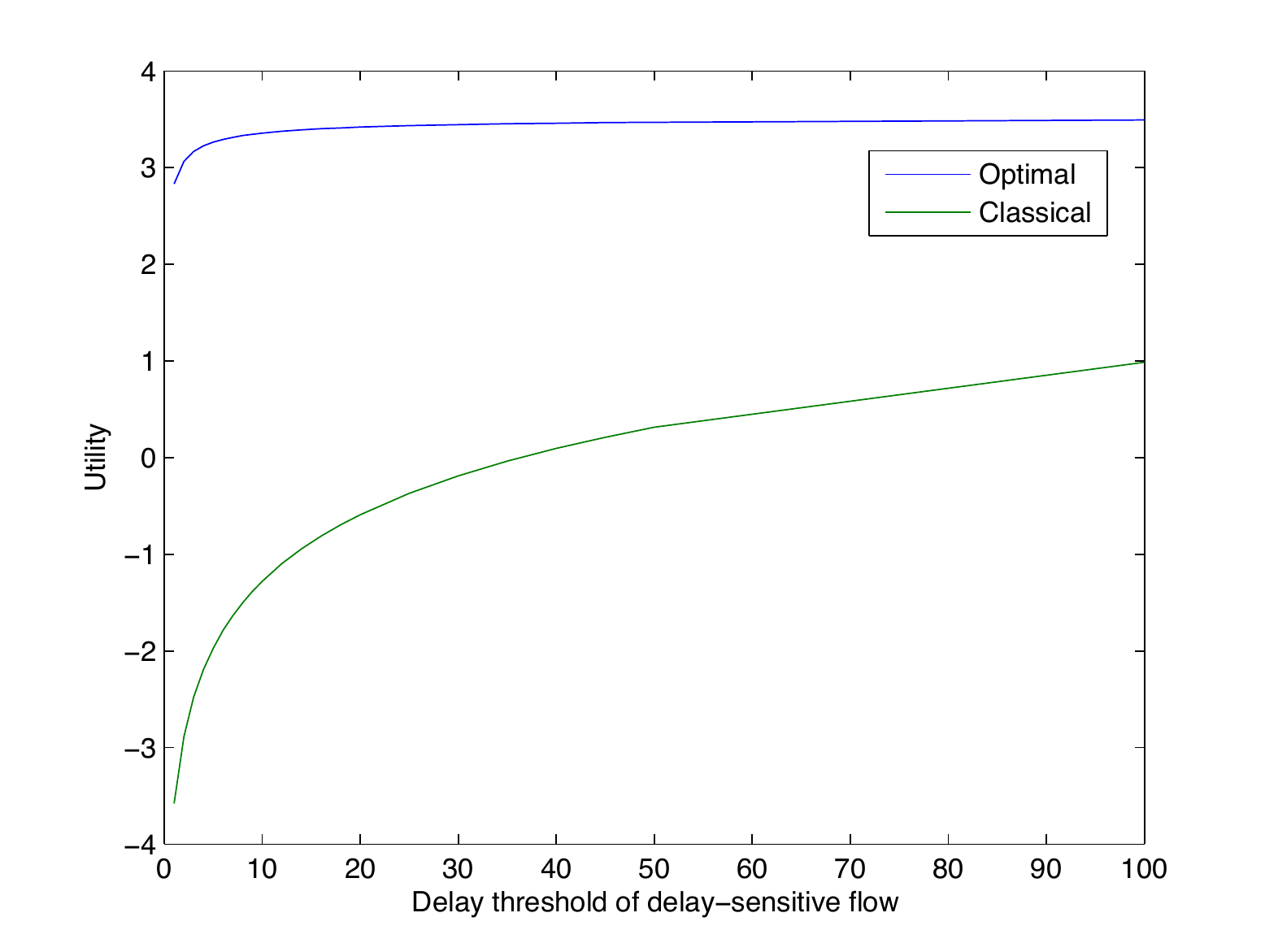}
}
\caption{An illustration of a single cell wireless LAN with 3 flows. The
central node represents the Access Point (AP), and the other nodes
represent the wireless STAtions.}
\label{fig:single_cell}
\end{figure}

However, we show that no such decomposition occurs when one or more
links are lossy or, one or more flows have finite delay deadlines.
Instead, in such cases, it is necessary to jointly optimise the flow
airtimes and coding rates. Further, we show that the resulting
allocation of airtime and coding rates is qualitatively different from
classical results. For example, consider a single hop wireless network
carrying three flows, see Figure \ref{fig:single_cell}. Flow $f_1$ is a
delay--sensitive flow (e.g. video) while flows $f_2$ and $f_3$ are
delay--insensitive flows (e.g. TCP data). Transmissions are scheduled in
a TDMA manner, and the delay deadline for flow $f_1$ is one schedule
period, while the delay deadline for flows $f_2$ and $f_3$ is infinite.
The channel symbol error rate is $10^{-2}$ for all flows, and flows use
MDS codes for error correction. The proportionally fair airtime and
coding rate allocation that we show in this paper (see
Eqns.~\eqref{eqn:kkt-n} and \eqref{eqn:kkt-r}) results in the allocation
of  41\% of the airtime for flow $f_1$ while flow $f_2$ and flow $f_3$
each receive 29.5\%. Observe that the proportionally fair allocation
assigns unequal airtimes to the flows, which is a notable departure from
the usual equal--airtime property of the proportional fair allocation
when selection of delay deadlines and coding rate are not included, e.g.
see \cite{eq_air_time}.  The optimal coding rate is 0.62 for flow $f_1$
and 0.97 for flows $f_2$ and $f_3$. The coding rate for flow $f_1$ is
much lower than for flows $f_2$ and $f_3$, since a smaller block size
must be used by flow $f_1$ (and more redundant symbols for
error--recovery) in order to respect the delay deadline. Due to the
delay deadline, these optimal coding rates yield non--zero loss rates.
For flow $f_1$, the packet loss rate at the receiver, after decoding, is
20\%, whereas flows $f_2$ and $f_3$ are loss free. This highlights an
important feature of the joint airtime and coding rate utility
optimisation. Namely, that it allows the throughput/loss/delay
trade--off amongst flows sharing network resources to be performed in a
principled, fair manner. Without consideration of coding rate, the
trade--off between throughput and loss cannot be fully understood or
optimally managed. Without consideration of airtime, the contention
between flows for shared network resources cannot be fully captured.

Proportional fairness can be formulated as a utility maximisation task,
with the utility being the sum of log flow rates.  Figure~1(b) compares
the optimal network utility with that obtained with a classical type of
approach where all flows are allocated equal airtime and the coding
rates are chosen based on the channel error probabilities alone (this
corresponds to ignoring the delay deadline of flow $f_1$).  It can be
seen that the optimal approach that we present in this paper potentially
offers significant performance benefits over classical methods.

We note that one of the reasons why the joint selection of
airtime/coding rate has not been previously studied is that the
proportional fair utility optimisation is non--convex, and hence,
powerful tools from convex optimisation cannot be applied directly.
Also, the study of the throughput performance by jointly considering the
coding and the MAC has not been performed before. One of the technical
contributions of this paper is to show that the proportional fair utility
optimisation can nevertheless be decomposed into a sequence of convex
optimisation problems. The solution to this sequence of convex problems
is the unique solution to the original non--convex optimisation.
Moreover, this solution can be written in an explicit form thereby
yielding considerable insight into the nature of the proportional fair
airtime/coding rate allocation.

The rest of the paper is organised as follows. The related literature on
utility optimal resource allocation is discussed in
Section~\ref{sec:related_work}. Section~\ref{sec:network_model} defines
 the network model; in particular, we describe the mesh network
architecture, the traffic model, and the channel model. We also discuss
the transmission scheduling model, decoding delay deadline, and the
network constraints. In Section~\ref{sec:problem_formulation}, we obtain
a measure for the end--to--end packet decoding error, and describe the
throughput of the network. In
Section~\ref{sec:network_utility_optimisation}, we formulate a network
utility maximisation problem subject to constraints on the transmission
schedule lengths, and discuss the optimization framework. 
In Section~\ref{sec:two_special_cases}, we discuss two special cases
of networks: delay--insensitive and loss--free networks, and show that
the tasks of obtaining optimal airtimes and coding rates decouple in
these special cases. 
We discuss
the optimal airtime/coding solution with some examples in
Section~\ref{sec:discussion}. 
Finally we conclude in Section~\ref{sec:conclusions}. The proofs of
Lemmas and Theorems are provided in the Appendix.

\section{Related Work}
\label{sec:related_work}
We consider a multi--hop Network Utility Maximisation (NUM) problem with
deadline constraints and with a practical model for the PHY layer. By
means of channel coding, we try to recover a packet from the channel
errors. Having a low coding rate helps in recovering the packets, but at
the cost of a small fraction of payload, and at the cost of the
transmission airtimes of other flows. Thus, we consider the problem of
resource allocation that answers the following question: {\em how to
allocate throughput across competing flows with each flow seeing
different channel conditions and respecting the delay deadline}. 

The problem of Network Utility Maximisation (NUM) has been studied in
various contexts, with NUM as a network layering tool introduced in
\cite{chiang}. 

Much of the work on NUM is concerned with the
flow scheduling and throughput allocation that achieves the network
stability region.  This work focuses on throughput and largely ignores delay constraints.  Resource allocation problems from the viewpoint of network
control and stability is studied by Georgiadis et al. in
\cite{now_res_alloc}. Network flow scheduling problems are studied in
a utility optimal framework by Shakkottai and Srikant in
\cite{net-opt}. In all these works and the references therein, the
emphasis is on the MAC layers and above. In \cite{now_res_alloc}, an
energy optimal scheduling problem is studied in which the PHYsical layer
is also considered. 

Some recent work explicitly includes delay constraints in the utility optimisation.  In \cite{li-atilla}, Li and Eryilmaz studied the problem of end--to--end
delay constrained scheduling in multi--hop networks. They propose
algorithms based on Lyapunov drift minimisation and pricing, and  
show that by dynamically selecting service disciplines, the proposed
algorithms significantly outperform existing throughput--optimal
scheduling algorithms. In \cite{rsrikant}, Jaramillo and Srikant studied
a resource allocation problem in ad hoc networks with elastic and
inelastic traffic with deadlines for packet reception, and obtained 
joint congestion control and scheduling algorithm that maximises a
network utility. In this work the focus is on congestion control and
scheduling, with the PHYsical layer considered to be
error--free. 

A short, preliminary version of the work in the current paper was presented in \cite{bsc}.

\section{Network Model}
\label{sec:network_model}

\subsection{Cellular Mesh Architecture}
\label{subsec:cellular_mesh_architecture}
We consider networks consisting of a set of $C \geq 1$ cells, $\C =
\{1,2,\cdots,C\}$ which define the ``interference domains'' in the
network. We allow intra--cell interference (\emph{i.e} transmissions by
nodes within the same cell interfere) but assume that there is no
inter--cell interference. This captures, for example, common network
architectures where nodes within a given cell use the same radio channel
while neighbouring cells using orthogonal radio channels. Within each
cell, any two nodes are within the decoding range of each other, and
hence, can communicate with each other. The cells are interconnected
using multi--radio bridging nodes to create a multi--hop wireless
network. A multi--radio bridging node $i$ connecting the set of cells
$\B(i)=\{c_1,..,c_n\}\subset \C$ can be thought of as a set of $n$
single radio nodes, one in each cell, interconnected by a high--speed,
loss--free wired backplane.  See, for example,
Figure~\ref{fig:mesh_network}.
\begin{figure}
\centering
\includegraphics[width=0.8\columnwidth]{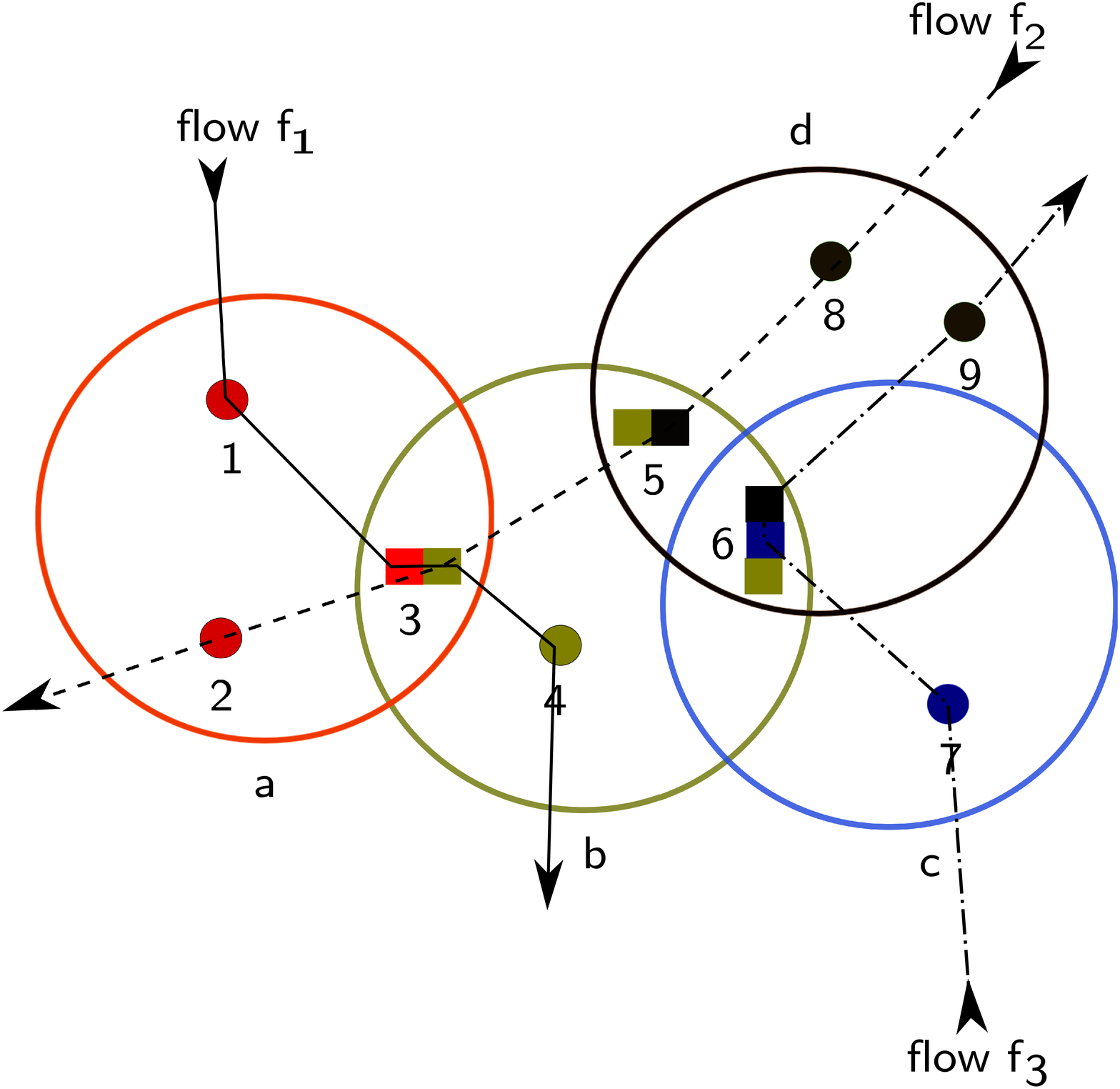}
\caption{{\bf An illustration of a wireless mesh network with 4 cells.}
Cells $a$, $b$, $c$, and $d$ use orthogonal channels CH$_1$, CH$_2$,
CH$_3$, and CH$_4$ respectively. Nodes 3, 5, and 6 are {\em bridge
nodes}. The bridge node 3 (resp. 5 and 6) is provided a time slice of
each of the channels CH$_1$ \& CH$_2$ (resp. CH$_2$ \& CH$_4$ for node 5
and CH$_2$\& CH$_3$\& CH$_4$ for node 6). Three flows $f_1, f_2$, and
$f_3$ are considered. In this example, $\C_{f_1} = \{a,b\}$, $\C_{f_2} =
\{d,b,a\}$, and $\C_{f_3} = \{c,d\}$.} 
\label{fig:mesh_network}
\end{figure}

\subsection{Unicast Flows}
\label{subsec:unicast_flows}
Data is transmitted across this multi--hop network as a set $\F$ $=$
$\{1,2,\cdots,F\}$, $F\geq 1$ of unicast flows. The route of each flow
$f$ $\in \F$ is given by $\C_f$ $=$ $\{c_1(f), c_2(f), \cdots,
c_{\ell_f}(f)\}$, where the source node $s(f) \in c_1(f)$ and the
destination node $d(f) \in c_{\ell_f}(f)$. We assume loop--free flows
(i.e., no two cells in $\C_f$ are same).  

\subsection{Binary Symmetric Channels}
\label{subsec:binary_symmetric_channels}
We associate a binary random variable $E_{f,c}[b]$ with the $b$'th bit
transmitted by flow $f$ in cell $c$. $E_{f,c}[b]=0$ indicates that the
bit is received correctly, and $E_{f,c}[k]=1$ indicates that the bit is
received incorrectly, i.e., the bit is ``flipped''. We assume that
$E_{f,c}[1], E_{f,c}[2],\cdots$ are independent and identically
distributed (iid), and ${\mathbb P}\{E_{f,c}[b]=1\}=\alpha_{f,c} \in
[0,0.5)$. That is, we have a binary symmetric channel with cross-over
probability $\alpha_{f,c}$. A transmitted bit may be ``flipped''
multiple times as it travels along the route of flow $f$, and is
received incorrectly at the flow destination only if there is an odd
number of such flips.  The end-to-end cross-over probability along the
route of flow $f$ is therefore given by 
\begin{eqnarray*}
\alpha_f = \sum_{\{ x_c \in \{0,1\}, c\in\C_f:\underset{c \in
\C_f}{\sum} x_c \ \text{is odd}\}} \ \ 
\prod_{c \in \C_f} \alpha_{f,c}^{x_c}
\ \left(1-\alpha_{f,c}\right)^{1-x_c}.
\end{eqnarray*}
Note that we can accommodate transmission of symbols from any $2^m =
M$--ary alphabet (i.e. not just transmission of binary symbols) by
associating $m$ channel uses of the BSC for every transmitted symbol.
The symbol error probability (for any $m \geq 1$) is then given by
$\beta_f = 1 - (1-\alpha_f)^{m}$.

In this channel model, the channel processes across time are independent
copies of the BSCs. In practice this can be realised by means of an
interleaver of sufficient depth (after the channel encoder), which
randomly shuffles the encoded symbols, combined with a de-interleaver
(before the channel decoder) at the receiver. This interleaving and
de--interleaving randomly mixes any channel fades, which can then be
modelled as independent channel processes across time. 

\subsection{Flow Transmission Scheduling}
\label{subsec:flow_transmission_scheduling}
A scheduler assigns a time slice of duration $T_{f,c} > 0$ time units to
each flow $f$ that flows through cell $c$, subject to the constraint
that $\sum_{f:c\in \C_f} T_{f,c}\le T_c$ where $T_c$ is the period of
the schedule in cell $c$. We consider a periodic scheduling strategy in
which, in each cell $c$, service is given to the flows in a round robin
fashion, and that each flow $f$ in cell $c$ gets a time slice of
$T_{f,c}$ units in every schedule. 

\subsection{Flow Decoding Delay Deadline}
\label{subsec:flow_decoding_delay_deadline}
At the source node $s(f)$ for flow $f$, we assume that $k_f$ symbols
arrive in each time slot, which allows us to simplify the analysis by
ignoring queueing. Information symbols are formed into blocks of $D_f
k_f$ symbols, where $D_f \in \{1,2,3,\cdots\}$ is the number of time
slots that the block may span. Each block of $D_f k_f$ information
symbols is encoded into a block of $D_f n_f$ coded symbols, where $n_f =
k_f/r_f$ symbols, with coding rate $0< {r_f}\le1$. Here, $n_f$ is the
number of encoded symbols transmitted in one slot i.e. the transmitted
packet size. The code employed for encoding is discussed in
Section~\ref{sec:problem_formulation}. The quantity $D_f$ is a user
or operator supplied quality of service parameter. It specifies the
decoding delay deadline for flow $f$, since after the flow destination
has collected at most $D_f$ successive coded packets it must attempt to
decode the encoded information symbols.

\subsection{Network Constraints on Coding Rate}
\label{subsec:network_constraints_on_coding_rate}
For flow $f$ in cell $c$, let $w_{f,c}$ be the rate of transmission in
symbols/second, which is determined by the modulation and spectral bandwidth
used for signal transmission and the within-cell FEC used. Each cell $c \in \C_f$ along the route of
flow $f$ allocates an airtime of at least $\frac{n_f}{w_{f,c}}$ in order to
transmit the packets of flow $f$. Let  $\F_c := \{f \in \F: c \in
\C_f\}$ be the set of flows that are routed through cell $c$. We recall
that the transmissions in any cell $c$ are scheduled in a TDMA fashion,
and hence, the total time required for transmitting packets for all
flows in cell $c$ is given by $\sum_{f \in \F_c} \frac{n_f}{w_{f,c}}$.
Since, for cell $c$, the transmission schedule interval is $T_c$ units
of time, the encoded packet size $n_f$ must satisfy the schedulability
constraint $$ \sum_{f \in \F_c} \frac{n_f}{w_{f,c}} \leqslant T_c $$
Note that since we require sufficient transmit time at each cell along
route $\C_f$ to allow $n_f$ coded symbols to be transmitted in every
schedule period, there is no queueing at the cells along the route of a
flow.

\section{Packet Error Probability}
\label{sec:problem_formulation}
Each transmitted symbol of flow $f$ reaches the destination node
erroneously with probability $\beta_f$. Hence, to help protect against
errors when recovering the information symbols, we encode information
symbols at the source nodes using a block code (we note here that a
convolutional code with zero--padding is also a block code). An
$(n,k,d)$ block code has the following properties. The encoder takes a
sequence of $k$ information symbols as input, and generates a sequence
of $n \geq k$ coded symbols as output. The decoder takes a sequence of
$n$ coded symbols as input, and outputs a sequence of $k$ information
symbols. These information symbols will be error--free provided no more
than  $\lfloor \frac{d-1}{2}\rfloor$ of the coded symbols are corrupted.
The Singleton bound \cite{mc_williams_sloane} tells us that $d \leqslant
n - k +1$, with equality for maximum--distance separable (MDS) codes.
Thus, an MDS code can correct up to
\begin{align}\label{eq:mdsbound}
\left\lfloor\frac{d-1}{2}\right\rfloor =\left\lfloor\frac{n-k}{2}\right\rfloor
\end{align}
errors. Examples for MDS codes include Reed--Solomon codes
\cite{mc_williams_sloane}, and MDS--convolutional codes
\cite{mds_conv_codes}. In \cite{mds_conv_codes}, the authors show the
existence of MDS--convolutional codes for any code rate. Hereafter, we
will make use of Eqn.~\eqref{eq:mdsbound}, and so, confine consideration to MDS
codes. However, the analysis can be readily extended to other types of
code provided a corresponding bound on $d$ is available. 

Consider a coded block of flow $f$ and let $i\in \{1, 2, \cdots,
D_fn_f\}$ index the symbols in the block. Let $E_f[i]$ be a binary random
variable which equals $0$ when the $i$'th coded symbol is received
correctly and which equals $1$ when it is received corrupted. ${\mathbb
P}\left\{ E_f[i] = 1\right\} = \beta_f$ and  ${\mathbb P}\left\{ E_f[i] =
0\right\} = 1-\beta_f$. From Eqn.~\eqref{eq:mdsbound}, the probability
of the block being decoded incorrectly is given by
$$
\widetilde{e}_f = 
{\mathbb P}\left\{\sum_{i=1}^{D_f n_f}E_f[i] > \frac{D_f n_f- D_fk_f}{2} \right\}
$$
The symbol errors $E_f[1], \ E_f[2], \ \cdots, \ E_f[D_fn_f]$ are i.i.d.
Bernoulli random variables, and so, the $\sum_{i=1}^{D_fn_f}E_f[i]$ is a
binomial random variable. Hence, the probability of a decoding error can
be computed exactly. However, the exact expression is combinatorial in
nature, and is not tractable for further analysis. We therefore proceed
by obtaining upper and lower bounds on the error probability, and show
that the bounds are the same up to a prefactor, and that the prefactor
decreases as the block size $D_fn_f$ increases. Hence, we pose the NUM
based on the upper bound on the error probability. Also, we relax the
following constraints: $n_f \in \mathbb{Z}_+$ and $k_f \in
\mathbb{Z}_+$, and allow them to take positive real values, i.e.,
$n_f \in \mathbb{R}_+$ and $k_f \in \mathbb{R}_+$.

\subsection{Upper and Lower Bounds}
\label{subsec:upper_and_lower_bounds}
\begin{lemma}[Upper Bound]
\label{lem:pe_bound}
The end--to--end probability $\widetilde{e}_f$ of a decoding error for
flow $f$ satisfies
\begin{align}
\label{eqn:pe_bound}
\widetilde{e}_f &\leq \ \exp\left(-D_f n_f I_{E_f[1]}\left(x_f;\theta_f\right)\right)\\
  &=:   \ e_f(\theta_f,n_f,r_f).\nonumber
\end{align}
where $x_f := \frac{1-r_f}{2}$, $r_f=k_f/n_f$ is the coding rate,
$\theta_f >0$ is the Chernoff--bound parameter and the function
$I_Z(x;\theta) := \theta x - \ln(\myexp{e^{\theta Z}})$ is called the
rate function in large deviations theory.
\end{lemma}
\begin{proof}
See Appendix~\ref{app:proof_lemma_bound_pe}.
\end{proof}

\begin{lemma}[Lower Bound]
\label{lem:pe_lower_bound}
The end--to--end probability $\widetilde{e}_f$ of a decoding error for
flow $f$ satisfies
\begin{align}
\widetilde{e}_f &\ge  
\Gamma \exp\left(-{D_fn_f} I\left({\cal
 B}\left(x_f\right)\|{\cal B}(\beta_f)\right)\right)
\end{align}
where
$$
\Gamma=\frac{\beta_f}{1-\beta_f}  \exp\left(-{D_fn_f} H\left({\cal B}\left(x_f\right)\right)\right)
$$
and $x_f := \frac{1-r_f}{2}$, ${\cal B}(x)$ is the Bernoulli
distribution with parameter $x$, $H({\cal P})$ is the entropy of
probability mass function (pmf) ${\cal P}$, and $I({\cal P}\|{\cal Q})$
is the information divergence between the pmfs ${\cal P}$ and ${\cal
Q}$. 
\end{lemma}
\begin{proof}
See Appendix~\ref{app:proof_lemma_lower_bound_pe}.
\end{proof}

\subsection{Tightness of Bounds}
\label{subsec:tightness_of_bounds}
It can be verified that 
\begin{align*}
I_{E_f[1]}\left(x_f;\theta_f\right) &= \theta_f x_f - \ln\left(1-\beta_f+\beta_f e^{\theta_f}\right)
\end{align*}
Since $\theta_f>0$ is a free parameter, we can select the value that
maximises $I_{E_f[1]}\left(x_f;\theta_f\right)$ and so provides the
tightest upper bound. It can be verified (e.g. by inspection of the
second derivative) that $I_{E_f[1]}\left(x_f;\theta_f\right)$ is concave
in $\theta_f$ and so the KKT conditions are necessary and sufficient for
an optimum. The KKT condition here is
\begin{align*}
\frac{\partial I_{E_f[1]}\left(x_f;\theta_f\right)}{\partial\theta_f}
&=   x_f - \frac{\beta_fe^{\theta_f}}{1-\beta_f+\beta_fe^{\theta_f}}
=0
\end{align*}
which is solved by
\begin{align*}
 \theta_f^* & = \ln\left(\frac{x_f}{\beta_f}\right) -
 \ln\left(\frac{1-x_f}{1-\beta_f}\right). 
\end{align*}
provided $x_f>\beta_f$. Substituting for $\theta_f^*$,
\begin{align*}
\max_{\theta_f > 0} I_{E_f[1]}\left(x_f;\theta_f\right) &=I_{E_f[1]}\left(x_f;\theta_f^*\right)\\
&=x_f \ln\left(\frac{x_f}{\beta_f}\right) +(1-x_f) \ln\left(\frac{1-x_f}{1-\beta_f}\right)\\
&=I\left({\cal B}\left(x_f\right)\|{\cal B}(\beta_f)\right)
\end{align*}
and by Lemmas~\ref{lem:pe_bound} and \ref{lem:pe_lower_bound}, the
probability $\widetilde{e}_f$ of a decoding error satisfies \newline
\begin{align*}
 \Gamma e^{-{D_fn_f} I\left({\cal
 B}\left(x_f\right)\|{\cal B}(\beta_f)\right)}
\le \widetilde{e}_f  \le e^{-D_f n_f I\left({\cal
B}\left(x_f\right)\|{\cal B}(\beta_f)\right)}.
\end{align*}
It can be seen that the upper and lower bounds are the same to within
prefactor $\Gamma$, and  the gap between these bounds decreases
exponentially as the block size $D_f n_f$ increases.

\section{Network Utility Optimisation} 
\label{sec:network_utility_optimisation} 
We are interested in the fair allocation of flow airtimes and coding
rates amongst flows in the network. Other things being equal, we expect
that decreasing the coding rate $r_f$ (i.e., increasing the number 
$D_fn_f - D_fk_f$ of redundant symbols transmitted) for flow $f$ will
decrease the error probability $e_f$, and so increases the flow throughput.
However, decreasing the coding rate increases the coded packet size
$D_fn_f$, and so increases the airtime used by flow $f$. Since the network
capacity is limited and shared by other flows, this generally decreases
the airtime available to other flows and so decreases their throughout. 
Similarly, increasing the packet size $D_fk_f$ of flow $f$ increases its
throughput but at the cost of increased airtime and a reduction in the
throughput of other flows. We formulate this tradeoff as a utility fair 
optimisation problem. In particular, we focus on the proportional fair
allocation since it is of wide interest and, as we will see, is
tractable, despite the non--convex nature of the optimisation.

The utility fair optimisation problem is\newline
\begin{minipage}{\columnwidth}
\begin{align}
\max_{{\bm \theta},({\bm n},{\bm x})} \ \ \ & 
U\left({\bm \theta}, ({\bm n}, {\bm x})\right) \label{eq:P1}\\ 
\text{subject to} \quad
&\underset{f: c \in \C_f  }{\sum} \frac{n_f}{w_{f,c}}  \leq  T_c,  &&\forall c \in \C \label{eq:cons1} \\ 
 &\theta_f  >  0, &&\forall f \in \F                  \label{eq:cons2}     \\
 & x_f  \leq  \overline{\lambda}_f \,  &&\forall f \in \F  \label{eq:cons3a} \\
 & x_f  \geq  \underline{\lambda}_f \, &&\forall f \in \F  \label{eq:cons3b} 
\end{align}
\normalsize
\end{minipage}
\newline
with ${\bm \theta}:=[\theta_f]_{f \in \F}$ the vector of Chernoff
parameters, ${\bm n} := [n_f]_{f\in \F}$ the vector of flow packet
sizes, and ${\bm x} := [x_f]_{f \in \F}$ the vector of flow coding rates
(where we recall that $x_f = (1-r_f)/2$). Eqn.~\eqref{eq:cons1} enforces
the network capacity (or the flow schedulability) constraints,
Eqn.~\eqref{eq:cons2} the positivity constraint on the Chernoff
parameters, and the constraints Eqns.~\eqref{eq:cons3a}--\eqref{eq:cons3b} are
introduced for technical reasons that will be discussed in more detail
shortly.

For proportional fairness, we select the sum of the log of the flow
throughputs as our network utility $U$. For flow $f$ the expected throughput is
$D_fk_f(1-\widetilde{e}_f)$ symbols in every time interval of duration $D_f
T_{d(f)}$ (we recall that $d(f)$ is the destination cell of flow $f$),
which is the same as $k_f(1-\widetilde{e}_f)$ symbols every time interval of
duration $ T_{d(f)}$, where $D_f k_f$ is the information packet size and
$\widetilde{e}_f$ the packet decoding error probability. As the exact
expression of $\widetilde{e}_f$ is intractable, we use the
upper bound for $\widetilde{e}_f$ which is $e_f$. Thus, the objective
function is given by 
\begin{align*}
U\left({\bm \theta}, ({\bm n}, {\bm x})\right)
& :=  \sum_{f \in \F} \ln\left(k_f \left(1 - e_f(\theta_f,n_f,x_f)\right)\right)\\
& =  \sum_{f \in \F} \ln\left(n_f r_f \left(1 - e_f(\theta_f,n_f,x_f)\right)\right)\\
& =  \sum_{f \in \F} \ln\left(n_f (1-2x_f) \left(1 -
e_f(\theta_f,n_f,x_f)\right)\right)\\
&  =  \sum_{f \in \F} \ln\left(n_f\right) + 
       \sum_{f \in \F} \ln\left(1-2x_f\right) \\ & \quad +
       \sum_{f \in \F} \ln \left(1 - e_f(\theta_f,n_f,x_f)\right).
\end{align*}
The optimisation problem yields the proportional fair flow coding rates
and coded packet size $n_f$. Since the PHY transmission rates $w_{f,c}$
are known parameters, the coded packet size is proportional to the
airtime used by a flow (i.e., the airtime is given by $n_f/w_{f,c}$).

\subsection{Non--Convexity}
\label{subsec:non_convexity}
The objective function $U({\bm \theta},({\bm n},{\bm x}))$ is separable
in $(\theta_f,(n_f,x_f))$ for each flow $f$. However, it can be readily
verified that $\ln\left(1-e_f(\theta_f,n_f,x_f)\right)$ is not jointly
concave in $(\theta_f,(n_f,x_f))$, and so, the optimisation is
non--convex. Hence, the network utility maximisation problem defined in
Eqns.~\eqref{eq:P1}--\eqref{eq:cons3b} is not in the standard convex
optimisation framework.

\subsection{Reformulation as Sequential Optimisations}
 We proceed by making the following key observation.
\begin{lemma}
\label{lem:opt_joint_vs_sep}.
For convex sets ${\cal Y}$ and ${\cal Z}$, and for a function $f:{\cal
Y}\times{\cal Z}\to {\mathbb R}$ that is concave in $y \in {\cal Y}$ and
in $z \in {\cal Z}$, but not jointly in $(y,z)$, the solution to the joint
optimisation problem 
\newline
\begin{minipage}{\columnwidth}
\begin{align}
\label{eqn:joint_opt_prob}
\max_{y \in {\cal Y}, z \in {\cal Z}} f(y,z) 
\end{align}
\end{minipage}
is unique, and is the same as the solution to
\newline
\begin{minipage}{\columnwidth}
\begin{align}
\max_{z \in {\cal Z}} \max_{y \in {\cal Y}} f(y,z),\\\nonumber 
\end{align}
\end{minipage}
\newline
if $f(y^*(z),z)$ is a concave function of $z$, where for each $z \in
{\cal Z}$, $y^*(z) := \underset{y \in {\cal
	Y}}{\arg\max} f(y,z)$.
\end{lemma}
\begin{proof}
See Appendix~\ref{app:proof_lemma_opt_joint_vs_sep}.
\end{proof}

This lemma establishes conditions under which we can transform a
non--convex optimisation into a sequence of convex optimisations.
Roughly speaking, we proceed by optimising over each variable in turn
and substituting the optimal variable value that is found back into the
objective function. This creates a sequence of objective functions.
Provided each member of this sequence is concave in the variable being
optimised (but not necessary jointly concave in all variables), the
solution to the sequence of convex optimisations coincides with the
solution to the original non--convex optimisation.  Evidently, the
condition that concavity holds for every objective function in this
sequence is extremely strong. Remarkably, however, we show that it is
satisfied in our present network utility optimisation.

\subsection{Optimal ${\bm \theta}^*(x_f)$}
\label{subsec:opt_theta}
Taking a sequential optimisation approach, we begin by first solving the optimisation 
\begin{align*}
\max_{{\bm \theta}} \ \ \ & 
U\left({\bm \theta}, ({\bm n}, {\bm x})\right) \nn
\text{subject to} \quad
 &\theta_f  >  0, &&\forall f \in \F                    
\end{align*}
given packet sizes ${\bm n} \in \mathbb{Z}_+^F$ and coding rates ${\bm
x}\in [\underline{\lambda}_f,\overline{\lambda}_f]^F$.   The objective
function is separable and concave in the $\theta_f$s.  The partial
derivative of $U\left({\bm \theta}, ({\bm n}, {\bm x})\right)$ with
respect to $\theta_f$ is given by 
\begin{align}\label{eq:deriv}
\frac{\partial U\left({\bm \theta}, ({\bm n}, {\bm x})\right)}{\partial\theta_f}
&= \frac{e_fD_f n_f}{1-e_f} \left[ x_f - \frac{\beta_f e^{\theta_f}}{1-\beta_f+\beta_fe^{\theta_f}}\right]
\end{align}
Setting this derivative equal to zero, provided $x_f> \beta_f$ this is
solved by
\begin{align}
\frac{\beta_f e^{\theta_f^*}}{1-\beta_f+\beta_fe^{\theta_f^*}} &=  x_f \nn
\text{or},  e^{\theta_f^*} &= \frac{x_f}{\beta_f}\frac{1-\beta_f}{1-x_f} \nn
\text{or},       \theta_f^*&=  \ln\left(\frac{x_f}{\beta_f}\right) -
 \ln\left(\frac{1-x_f}{1-\beta_f}\right). 
 \label{eq:opt_theta_star}
\end{align}
Observe that in fact $\theta_f^*$ is function only of $x_f$ and not both
$n_f$ and $x_f$. The requirement  for $x_f> \beta_f$ ensures that
$\theta_f^*>0$. When $x_f\le \beta_f$, the derivative
(Eqn.~\eqref{eq:deriv})
is negative for all $\theta_f >0$. In this case, the optimum
$\theta_f^*$ is zero which yields an error probability $e_f$ of one.
Thus, for error recovery we require $x_f> \beta_f$ i.e. the coding rate
$r_f < 1-2\beta_f$,  and for a non--empty feasible
region in the NUM problem formulation in
Eqns.~\eqref{eq:P1}--\eqref{eq:cons3b}  the constraints on $x_f$ should satisfy the following:
$\overline{\lambda}_f < 1-2\beta_f$ and $\underline{\lambda}_f > 0$. We
note that the capacity region for a BSC having a cross--over probability
$\alpha_f$ with an $m$--ary signalling is $(0, m(1-H(\alpha_f)))$, and
the coding rate $1-2\beta_f$ lies in the capacity region.

\subsection{Optimal $(n_f^*,x_f^*)$}
\label{sec:convexity_condition}
The next step in our sequential optimisation approach is to solve
\begin{align*}
\max_{({\bm n},{\bm x})} \ \ \ & 
U\left( {\bm \theta^*}({\bm x}), ({\bm n}, {\bm x})\right) \nn
\text{subject to} \quad
&\underset{f: c \in \C_f  }{\sum} \frac{n_f}{w_{f,c}}  \leq  T_c,  &&\forall c \in \C  \\ 
 & x_f  \leq  \overline{\lambda}_f \,  &&\forall f \in \F   \\
 & x_f  \geq  \underline{\lambda}_f \, &&\forall f \in \F  
\end{align*}
That is, we substitute into the objective function for the optimal
$\theta_f^*$ found in Section \ref{subsec:opt_theta}. Defining $I_f =
x_f \ln\left(\frac{x_f}{\beta_f}\right) +(1-x_f)
\ln\left(\frac{1-x_f}{1-\beta_f}\right)$,
\begin{align*}
U\left( {\bm \theta^*}({\bm x}), ({\bm n}, {\bm x})\right)
&=\sum_{f \in \F} \ln\left(n_f\right) + 
       \sum_{f \in \F} \ln\left(1-2x_f\right) \\ & \quad +
       \sum_{f \in \F} \ln \left(1 - e^{-D_f n_f I_f}\right).
\end{align*}
It can be verified that $U\left( {\bm \theta^*}({\bm x}), ({\bm n}, 
{\bm x})\right)$ is not jointly concave in $({\bm n}, {\bm x})$. To
proceed, we therefore rewrite the objective in terms of the
log--transformed variables $\widetilde{n}_f = \ln(n_f)$ and
$\widetilde{I}_f =\ln(I_f)$. Observe that the mapping from $n_f$ to
$\widetilde{n}_f $ is invertible and similarly the mapping from $x_f$ to
$\widetilde{I}_f$. Since $\widetilde{I}_f$ is a monotone increasing
function of $x_f$ (this can be verified by inspection of the first
derivative), the inverse mapping from $\widetilde{I}_f$ to $x_f$ exists
and is one-to-one. With the obvious abuse of notation, we write inverse
map as $x_f(\widetilde{I}_f)$. In terms of these log--transformed
co-ordinates, the objective function is
$U\left( {\bm \theta^*}({\bm \tilde{I}}), ({\bm \tilde{n}}, {\bm
\tilde{I}})\right)$. We note that the problem defined in 
Eqns.~\eqref{eq:P1}--\eqref{eq:cons3b} is equivalent to the problem, 
\newline
\begin{minipage}{\columnwidth}
\begin{align}
\max_{{\bm \theta},(\tilde{\bm n},\tilde{\bm I})} \ \ \ & 
U\left( {\bm \theta}, ({\tilde{\bm n}}, 
{\tilde{\bm I}})\right) \nn
\text{subject to} \quad
&\underset{f: c \in \C_f  }{\sum} \frac{e^{\tilde{n}_f}}{w_{f,c}}  \leq
T_c,  &&\forall c \in \C \label{eq:consb1} \\ 
 &\theta_f  >  0, &&\forall f \in \F                  \label{eq:consb2}     \\
 & \tilde{I}_f  \leq  \overline{\tilde\lambda}_f \,  &&\forall f \in \F
\label{eq:consb3a} \\
 &\tilde{I}_f  \geq  \underline{\tilde\lambda}_f \, &&\forall f \in \F
\label{eq:consb3b} 
\end{align}
\normalsize
\end{minipage}
\newline
and hence, by Lemma~\ref{lem:opt_joint_vs_sep}, the solution to the
log--transformed problem is the same as that of the problem defined in
Eqns.~\eqref{eq:P1}--\eqref{eq:cons3b}. We solve the maximisation
problem by convex optimisation method. We show that the objective
function is jointly concave in $\left({\bm \tilde{n}}, {\bm
\tilde{I}}\right)$ in the following Lemma.
\begin{lemma}
\label{lemma:convexity_of_log_tx}
\begin{align*}
U\left( {\bm \theta^*}({\bm \tilde{I}}), ({\bm \tilde{n}}, {\bm \tilde{I}})\right)
&=\sum_{f \in \F} \tilde{n}_f + 
       \sum_{f \in \F} \ln\left(1-2x_f(\widetilde{I}_f)\right) \\ & \quad +
       \sum_{f \in \F} \ln \left(1 - e^{-D_f e^{\tilde{n}_f + \tilde{I}_f }}\right).
\end{align*}
is jointly concave in $\widetilde{n}_f$ and $\widetilde{I}_f$. 
\end{lemma}

\begin{proof}
See Appendix~\ref{app:proof_lemma_convexity_of_log_tx}.
\end{proof}

Hence, we have the following convex optimisation problem
\begin{align}
\label{eqn:max_x_problem}
\max_{({\bm \tilde{n}},{\bm \tilde{I}})} \ \ \ & 
U\left( {\bm \theta^*}({\bm \tilde{I}}), ({\bm \tilde{n}}, {\bm
\tilde{I}})\right) \\
\text{subject to} \quad
&\underset{f: c \in \C_f  }{\sum} \frac{e^{\widetilde{n}_f}}{w_{f,c}}
\leq  T_c,  &&\forall c \in \C\label{eq:log_const0}\\
&\tilde{I}_f  \leq  \overline{\tilde\lambda}_f \,  &&\forall f \in \F
\label{eq:log_const1}\\
&\tilde{I}_f  \geq  \underline{\tilde\lambda}_f \, &&\forall f \in \F
\label{eq:log_const2}
\end{align}
We solve the above maximisation problem using the Lagrangian relaxation
approach. The Lagrangian function of the problem is given by 
\begin{align*}
& L\left(\widetilde{\bm n}, \widetilde{\bm I}, \bm{p},
{\underline{\bm \nu}}, {\overline{\bm \nu}}\right) \nn
&:= \underset{f\in\F}{\sum}
U_f(\theta_f^*(\widetilde{I}_f),(\widetilde{n}_f,\widetilde{I}_f)) 
- \underset{c\in\C}{\sum}p_c\left(
\underset{f \in \F_c}{\sum}\frac{e^{{\widetilde n}_f}}{w_{f,c}} - T_c
\right)\nn
& - \underset{f\in\F}{\sum}\underline{\nu}_f\left( 
\tilde{I}_f -  \overline{\tilde\lambda}_f 
\right)
  + \underset{f\in\F}{\sum}\overline{\nu}_f \left( 
\tilde{I}_f -  \underline{\tilde\lambda}_f 
\right)
\end{align*}
where $\bm{p} \geq {\bm 0}$, ${\underline{\bm \nu}} \geq {\bm 0}$, and
${\overline{\bm \nu}} \geq {\bm 0}$ are
Lagrangian multipliers corresponding to the constraints given in  
Eqns.~\eqref{eq:log_const0}--\eqref{eq:log_const2}. The channel error
probabilities $\beta_f$s are strictly positive, and the channel coding
rates are always assumed to be in the interior of the feasibility region. 
Hence, the constraints for the channel coding rate given in
Eqns.~\eqref{eq:log_const1}, \eqref{eq:log_const2} are not active at
the optimal point, and the Lagrangian costs $\underline{\nu}_f$s and
${\overline{\nu}_f}$s are zero. Thus, the shadow costs
corresponding to these constraints will not appear in the Lagrangian
relaxation. 

Since the optimisation problem falls within convex optimisation framework, and
the Slater condition is satisfied, strong duality holds. Hence, 
the KKT conditions are necessary and sufficient for optimality.
Differentiating the Lagrangian with respect to 
$\widetilde{n}_f$
at 
$\widetilde{n}_f =
\widetilde{n}_f^*$,
and setting equal to zero yields the KKT condition 
\begin{align}
\label{eqn:kkt-n}
1 
 + \frac{D_f 
 { \exp(\widetilde{n}_f^* + \widetilde{I}_f^*)}
 e^{-D_f \exp(\widetilde{n}_f^* + \widetilde{I}_f^*)}
 }{1-
 e^{-D_f \exp(\widetilde{n}_f^* + \widetilde{I}_f^*)}
 }
 &= \sum_{c\in{\cal C}_f} \frac{p_c e^{\widetilde{n}_f^*}}{w_{f,c}}\nn 
\text{or}, \ \  1 
 + \frac{D_f n_f^* I_f(x_f^*) e^{-D_f n_f^*I_f(x_f^*)}}{1-e^{-D_f n_f^*
 I_f(x_f^*)}}
 &= \sum_{c\in{\cal C}_f} \frac{p_c n_f^*
 }{w_{f,c}}
 \end{align}
Similarly, the KKT condition for $\widetilde{I}_f^*$ is
 \begin{align}
\label{eqn:kkt-r}
\frac{D_f 
 { \exp(\widetilde{n}_f^* + \widetilde{I}_f^*)}
 e^{-D_f \exp(\widetilde{n}_f^* + \widetilde{I}_f^*)}
 }{1-
 e^{-D_f \exp(\widetilde{n}_f^* + \widetilde{I}_f^*)}
 }
 &= \frac{2}{1-2x_f^*} \frac{I_f^*}{\theta_f^*}\nn
\text{or}, \ \ \
  \frac{D_f n_f^* I_f(x_f^*) e^{-D_f n_f^*I_f(x_f^*)}}{1-e^{-D_f n_f^*
  I_f(x_f^*)}}
 &= \frac{2}{1-2x_f^*} \frac{I_f(x_f^*)}{\theta_f^*(x_f^*)}
 \end{align}

Combining eqns.~\eqref{eqn:kkt-n} and \eqref{eqn:kkt-r}, yields
 \begin{align}\label{eq:primalsoln}
\sum_{c\in{\cal C}_f} \frac{p_c n_f^*}{w_{f,c}} - 1 
&=   \frac{2}{1-2x_f^*} \frac{I_f(x_f^*)}{\theta_f^*(x_f^*)}.
\end{align}
Observe that the LHS is a function of $n_f^*$ and the RHS is a function of $x_f^*$.
Thus, the choice of packet size parameter $n_f^*$ and
coding rate parameter $x_f^*$ are in general coupled. 

\subsection{Distributed Algorithm for Solving Optimisation}
Given the values of the Lagrange multipliers ${\bm p}^*$, the solution
to Eqn.~\eqref{eq:primalsoln} specifies the optimal packet size and
coding rate. To complete the solution to the optimisation it therefore
remains to calculate the multipliers ${\bm p}^*$. These cannot be
obtained in closed form since their values reflect the network topology
and details of flow routing. However, they can be readily found in a
distributed manner using a standard subgradient approach.

We proceed as follows. The dual problem for the primal problem defined in 
Eqn.~\eqref{eqn:max_x_problem} is given by
\begin{eqnarray*}
\min_{{\bm p} \geq 0} & & D(\bm{p}),
\end{eqnarray*}
where the dual function $D(\bm{p})$ is given by
\begin{align}
D(\bm{p})
&= \max_{(\widetilde{\bm n},\widetilde{\bm I})} \ 
\sum_{f \in \F} \ U_f(\theta_f^*(\widetilde{I}_f),(\widetilde
n_f,\widetilde I_f)) \nn
&\quad + \underset{c\in\C}{\sum}p_c\left( T_c - \underset{f \in
\F_c}{\sum} \frac{e^{\widetilde n_f}}{w_{f,c}} \right) \label{eq:supremum}\\
&= \underset{f\in\F}{\sum} \ 
U_f\left( \theta_f^*(\widetilde I_f({\bm p})), (\widetilde n_f^*({\bm p}),
\widetilde I_f^*({\bm p})) \right)\nn
&\quad +\underset{c\in\C}{\sum}p_c\left( T_c - \underset{f \in \F_c}{\sum}
\frac{e^{\widetilde n_f^*({\bm p})}}{w_{f,c}} 
\right).\nonumber
 \label{eq:d_star}
\end{align}
\normalsize

\noindent
From Eqn.~\eqref{eq:supremum}, for any $(\widetilde{\bm
n},\widetilde{\bm I})$, 
\begin{align*}
D(\bm{p}) &\geq \underset{f\in\F}{\sum} U_f(\theta_f^*(\widetilde
I_f),(\widetilde n_f,\widetilde I_f)) +\underset{c\in\C}{\sum}p_c\left( T_c - \underset{f \in \F_c}{\sum}
\frac{e^{\widetilde n_f}}{w_{f,c}}
\right),
\end{align*}
\normalsize
and in particular, the dual function  $D({\bm p})$ is greater than that
for $\widetilde{I}_f= \widetilde{I}_f^*({\widetilde{\bm p}})$ for some arbitrary $\widetilde{\bm p}$, i.e.,
\begin{align}
D(\bm{p})
&\geq \underset{f\in\F}{\sum} \ U_f\left(
\theta_f^*(\widetilde I_f^*({\widetilde{\bm p}})), 
(
\widetilde n_f^*({\widetilde{\bm p}}), 
\widetilde I_f^*({\widetilde{\bm p}}) 
)
\right)\nn 
&\quad+\underset{c\in\C}{\sum}p_c\left( T_c - \underset{f \in \F_c}{\sum}
\frac{e^{\widetilde n_f^*({\widetilde{\bm p}})}}{w_{f,c}}
\right)\nn
&= D(\widetilde{\bm p})
+\underset{c\in\C}{\sum}\left(p_c-\widetilde{p}_c\right)\left( T_c - \underset{f \in \F_c}{\sum}
\frac{e^{\widetilde n_f^*({\widetilde{\bm p}})}}{w_{f,c}}
\right)
\end{align}
\normalsize
Thus, a sub--gradient of $D(\cdot)$ at any $\widetilde{\bm p}$ is given
by the vector 
\begin{align}
\left[ T_c - \underset{f \in \F_c}{\sum}
\frac{n_f^*({\widetilde{\bm p}})}{w_{f,c}}\right]_{c \in
\C},\nonumber
\end{align}
and the projected subgradient descent update is
\begin{eqnarray*}
p_c(i+1) = \left[p_c(i)-\gamma\cdot
\left(T_c - \underset{f \in \F_c}{\sum}
\frac{n_f^*({\bm p}(i))}{w_{f,c}} 
\right)\right]^+
\end{eqnarray*}
where $\gamma > 0$ is a sufficiently small stepsize, and $[f(x)]^+ := \max\{f(x),0\}$
ensures that the Lagrange multiplier
never goes negative (see \cite{nedich}).

The subgradient updates can carried out locally by each cell $c$ since
the update of $p_c$ only requires knowledge of the packet sizes
$n_f^*({\bm p}(i))$ of flows $f \in \F_c$ traversing cell $c$. Thus, at
the beginning of each iteration $i$, the flow source nodes choose their
packet sizes as $D_f n_f^*({\bm p}(i))$ and the coding rates as
$1-2x_f^*({\bm p}(i))$, and each cell computes its cost based on the
packet sizes (or equivalently the rates) of flows through it. The
updated costs along the route of each flow are then fed back to the
source nodes to compute the packet size and coding rate for the next
iteration.

Observe that the Lagrange multiplier $p_c$ can be interpreted as the
cost of transmitting traffic through cell $c$. The amount of service
time that is available is given by $\Delta = T_c - \underset{f \in
\F_c}{\sum} \frac{n_f^*({\bm p}(i))}{w_{f,c}}$. When $\Delta$ is
positive and large, then the Lagrangian cost $p_c$ decreases rapidly
(because the dual function $D(\cdot)$ is convex), and when $\Delta$ is negative, then the
Lagrangian cost $p_c$ increases rapidly to make $\Delta \geq 0$. We note
that the increase or decrease of $p_c$ between successive iterations is
proportional to $\Delta$, the amount of service time available. Thus,
the sub--gradient procedure provides a dynamic control scheme to balance
the network load.

The resulting distributed implementation of the joint airtime/coding
rate optimisation task is summarised in Algorithm \ref{algo:one}.
\begin{algorithm}
\begin{algorithmic}
\STATE \emph{Each cell $c$ runs:}
\LOOP
\STATE 1. $p_c(i+1) = \left[p_c(i)-\gamma\cdot \left(T_c - \underset{f
\in \F_c}{\sum} \frac{n_f^*({\bm p}(i))}{w_{f,c}} \right)\right]^+$
\ENDLOOP
\end{algorithmic}
\begin{algorithmic}
\STATE \emph{The source for each flow $f$ runs:}
\LOOP
\STATE 1. Measure $\sum_{c\in{\cal C}_f} \frac{p_c}{w_{f,c}}$, the
aggregate cost of using the cells along the route of flow $f$.  E.g. if
each cell updates the header of transmitted packets to reflect this sum,
it can then be echoed back to the source by the flow destination. 
\STATE 2. Find the unique packet size $n_f^*$ and coding rate $x_f^*$
that solve (\ref{eq:primalsoln}). Since there are only two variables,
a simple numerical search can be used.
\ENDLOOP
\end{algorithmic}
\caption{Distributed implementation of joint airtime/coding rate
optimisation.}\label{algo:one}
\end{algorithm}

\section{Two Special Cases}
\label{sec:two_special_cases}

\subsection{Delay--Insensitive Networks}
Suppose the delay deadline $D_f\rightarrow \infty$ for all flows. For
any positive bounded $n_f^*I_f^*$, i.e., $0 < n_f^*I_f^* < \infty$, the
LHS of Eqn.~\eqref{eqn:kkt-r} can be written as   
\begin{align}
\frac{D_f n_f^*I_f^* e^{-D_f n_f^*I_f^*} }{1- e^{-D_f n_f^*I_f^*}}
&= \frac{D_f n_f^*I_f^*  }{e^{D_f n_f^*I_f^*}-1} \nn
&= \frac{D_f n_f^*I_f^*  }{\sum_{j = 1}^\infty \frac{(D_f n_f^*I_f^*)^j}{j!} } \nn
&= \frac{1 }{1+\sum_{j = 2}^\infty \frac{(D_f n_f^*I_f^*)^{j-1}}{j!} } \nn
\therefore \lim_{D_f \to \infty}\frac{D_f n_f^*I_f^* e^{-D_f n_f^*I_f^*} }{1-
e^{-D_f n_f^*I_f^*}} &= 0.
\end{align}
Thus, the asymptotic optimal coding rate $x_f^*$ as the delay deadline
requirement $D_f \to \infty$ is the solution to
\begin{align}\label{eq:delayx}
\frac{2}{1-2x_f^*}\frac{I_f(x_f^*)}{\theta^*(x_f^*)} = 0.
\end{align}
Since $\beta_f < x_f < 1/2$, it is sufficient to find the solution to 
\begin{align}\label{eq:delayx-new}
\frac{I_f(x_f^*)}{\theta^*(x_f^*)} &= 0. \nn
\text{Note that} \ 
\lim_{x_f^* \to \beta_f} 
\frac{I_f(x_f^*)}{\theta^*(x_f^*)} &= 0, \nn
\text{and hence}, \ 
\lim_{D_f \to \infty} x_f^* &= \beta_f.
\end{align}
Since this is the limiting solution and $x_f^* > \beta_f$, one can use 
$x_f^* = \beta_f + \epsilon$ for some arbitrarily small $\epsilon > 0$. 
Similarly, from Eqn.~\eqref{eqn:kkt-n}, the asymptotic optimal packet 
size $n_f^*$ as $D_f \to \infty$ is
\begin{align}\label{eq:delayn}
n_f^*=\frac{1}{\sum_{c\in C_f}p_c/w_{f,c}}
\end{align}
where the multipliers $p_c$ are obtained, as before, by subgradient
descent
\begin{eqnarray}\label{eq:delayp}
p_c(i+1) = \left[p_c(i)-\gamma\cdot
\left(T_c - \underset{f \in \F_c}{\sum}
\frac{n_f^*({\bm p}(i))}{w_{f,c}}
\right)\right]^+
\end{eqnarray}
Observe that the optimal coding rate $x_f^* = \beta_f + \epsilon$ which
is given by the solution of Eqn.~\eqref{eq:delayx-new} is determined
solely by the channel error rate $\beta_f$ of flow $f$. It is therefore
completely independent of the other network properties. In particular,
it is independent of the packet size $n_f^*$ used, of the other flows
sharing the network and of the network topology. Conversely, observe
that the optimal packet size $n_f^*$ in Eqn.~\eqref{eq:delayn} and
Eqn.~\eqref{eq:delayp} is dependent on the network topology and flow
routes, but is completely independent of the error rate $\beta_f$ and
coding rate $x_f$. That is, in delay--insensitive networks, the joint
airtime/coding rate optimisation task breaks into separate optimal
airtime allocation and optimal coding rate allocation tasks which are
completely decoupled. Our optimisation therefore yields a MAC/PHY
layering, whereby airtime allocation/transmission scheduling is handled
by the MAC whereas coding rate selection is handled by the PHY, with no
cross-layer communication. It is important to note, however, that this
layering does not occur in networks where one or more flows have finite
delay--deadlines; see Section~\ref{sec:discussion} for a more detailed
discussion.

\subsection{Loss-Free Networks}\label{sec:lossfree}
Suppose the channel symbol error rate $\beta_f = 0$ for all flows. 
From Eqn.~\eqref{eq:opt_theta_star}, we observe that 
\begin{align}
\lim_{\beta_f \to 0} \theta_f^* &= \infty,
\end{align}
and this yields $e_f=0$ for all flows. The objective function in 
Eqn.~\eqref{eqn:max_x_problem} degenerates to $\sum_{f \in \F}
\ln(n_f(1-2x_f))$. We note that for any $x^*_f>\beta_f > 0$, 
as $\beta_f \downarrow 0$, $I_f(x_f^*) \to \infty$. Hence, the LHS of 
Eqn.~\eqref{eqn:kkt-r} becomes,
\begin{align}
\lim_{\beta_f \downarrow 0}  
  \frac{D_f n_f^* I_f(x_f^*) e^{-D_f n_f^*I_f(x_f^*)}}{1-e^{-D_f n_f^*
  I_f(x_f^*)}}
&= 0. 
\end{align}
In the same was as in Eqn.~\eqref{eq:delayx-new}, this limit can be achieved by $x_f^*=0$ (i.e. $r_f^*=1$). Similarly, 
the optimal packet size is 
$n_f^*=\frac{1}{\sum_{c\in \C_f}p_c/w_{f,c}}$.
This optimal packet size is identical to that for delay--insensitive
networks, see Eqn.~\eqref{eq:delayn}, and it can be verified that in fact it
corresponds to the classical proportional fair rate allocation for loss--free networks, as expected.

\section{Examples}
\label{sec:discussion}

\subsection{Single cell}\label{subsec:x_f_s}
We begin by considering network examples consisting of a single cell
carrying multiple flows.   The network topology is illustrated
schematically in Figure \ref{fig:single_cell} and might correspond, for
example, to a WLAN.   

\begin{figure}
\centering
\subfigure[Optimal airtime allocation vs delay deadline $D$, N=2.]{
\includegraphics[width=0.8\columnwidth]{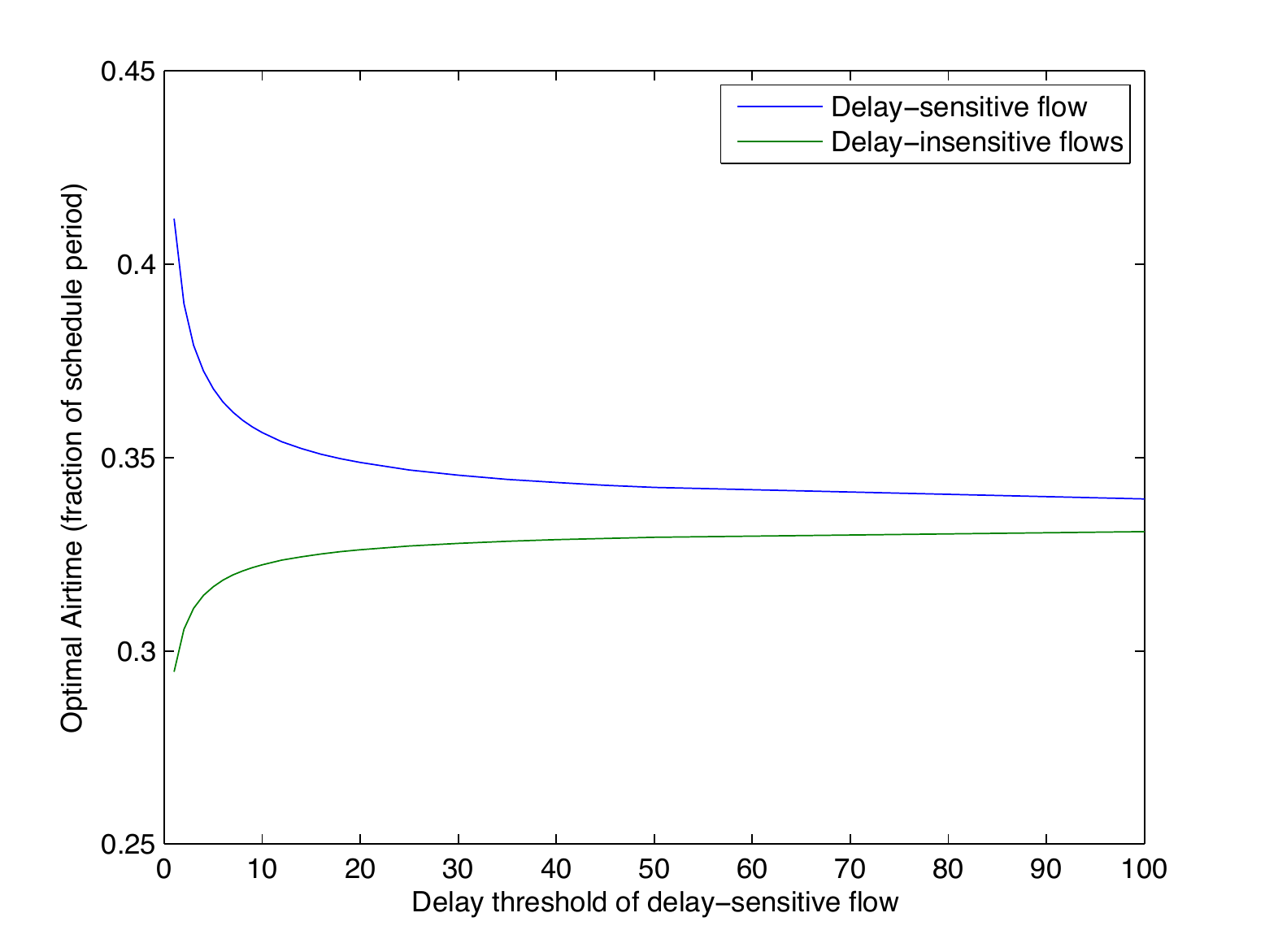}
}
\subfigure[Optimal airtime allocation vs N.  $D=1$.]{
\includegraphics[width=0.8\columnwidth]{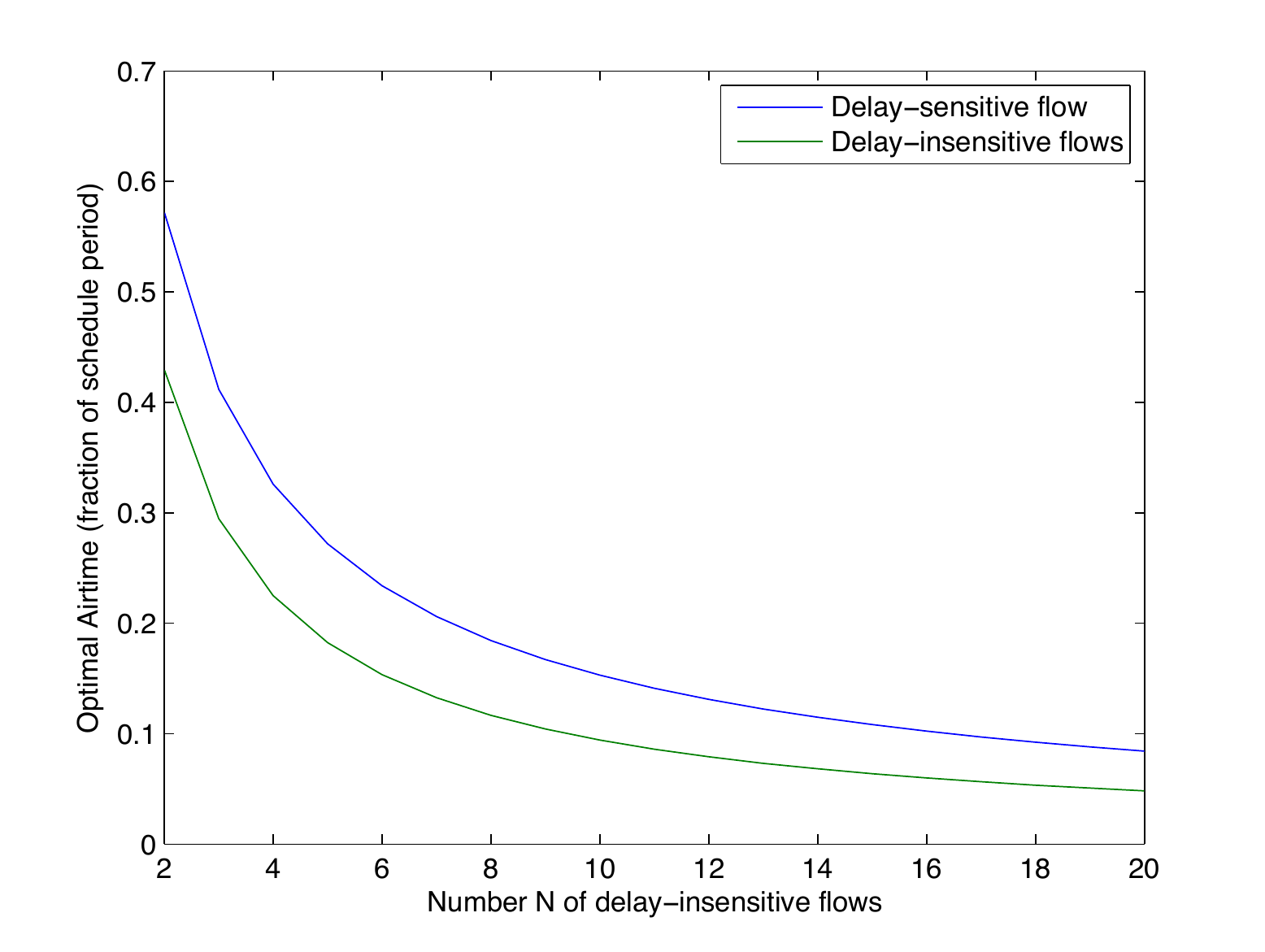}
}
\caption{Single WLAN with one delay--sensitive flow and $N$
delay--insensitive flows.  Delay sensitive flow has delay deadline $D$,
delay--insensitive flows have infinite delay deadlines.   Raw channel
symbol error rate is $10^{-2}$ for all flows, PHY rate for all flows is
10 symbols per schedule period.  Optimal airtimes are given as a
proportion of the schedule period.}
\label{fig:onecell1}
\end{figure}

\begin{figure}
\centering
\subfigure[Optimal airtime allocation vs channel symbol error rate for
flow 1, symbol error rate for flow 2 is held fixed at $10^{-2}$.]{
\includegraphics[width=0.8\columnwidth]{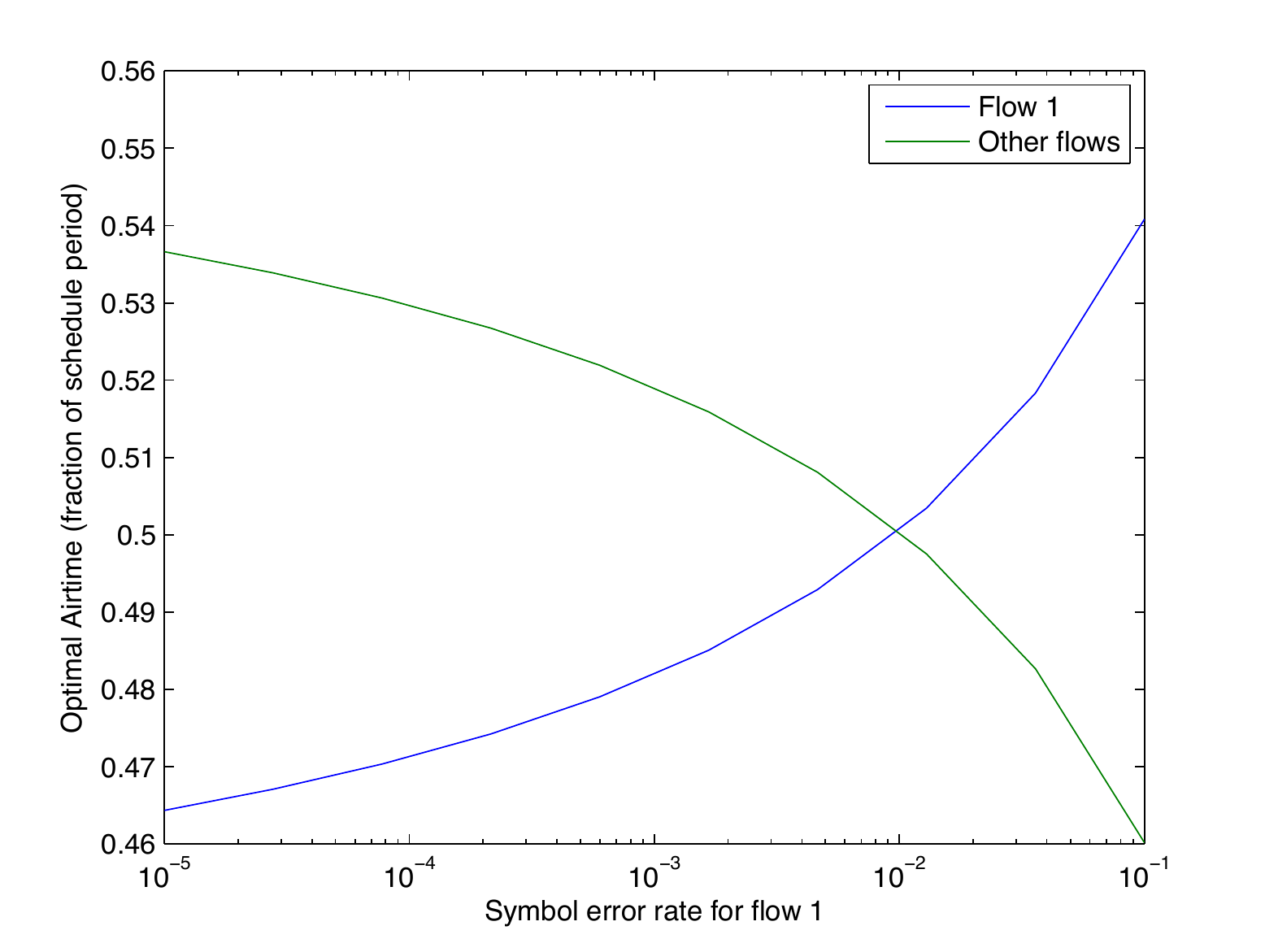}
}
\subfigure[Optimal airtime allocation vs PHY rate of flow 1, PHY rate
for flow 2 is held fixed at 10 symbols/schedule.]{
\includegraphics[width=0.8\columnwidth]{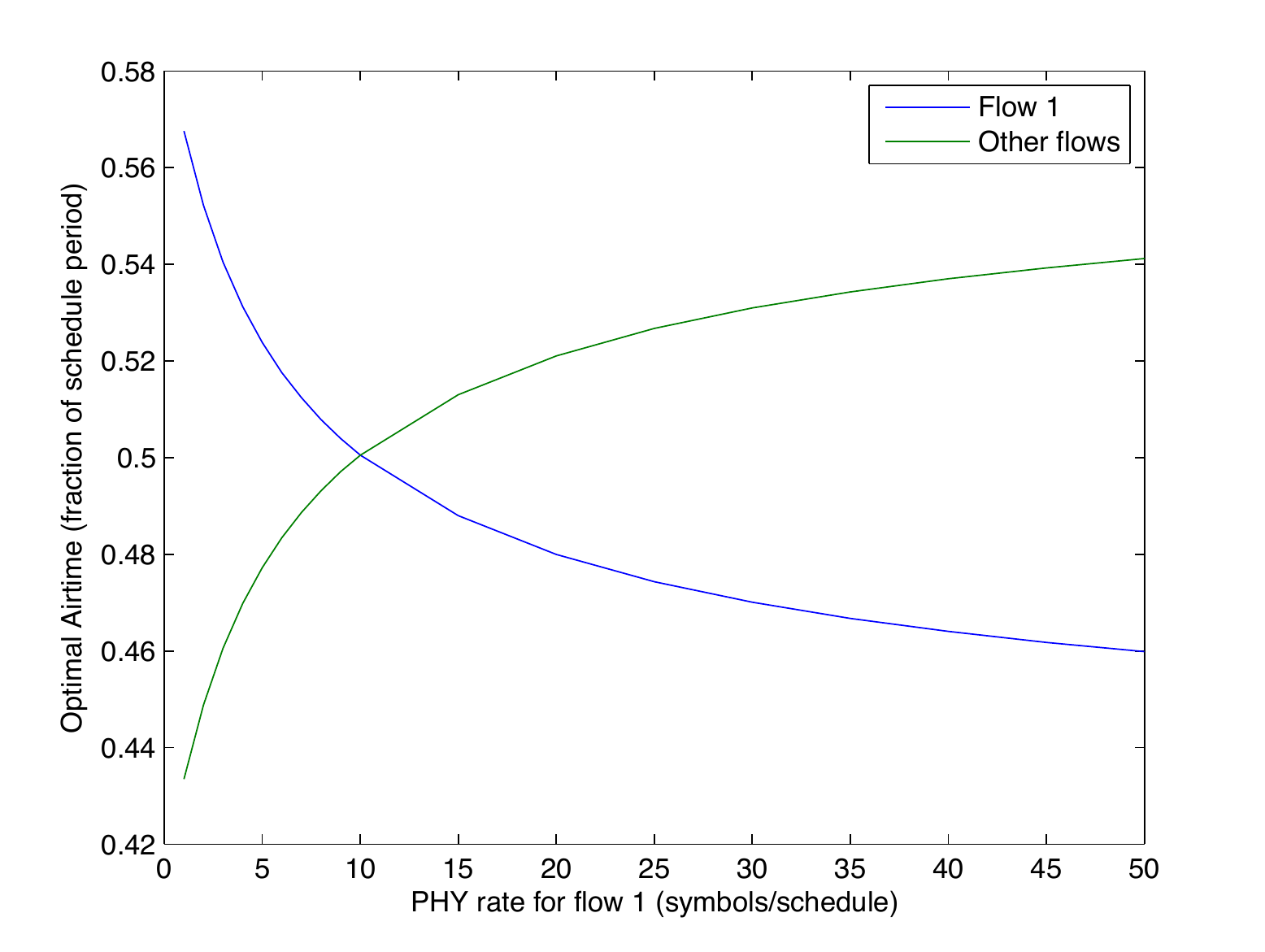}
}
\caption{Single WLAN with two delay--sensitive flows, both with delay
deadline $D=1$. In upper figure, PHY rate for both flows is 10 symbols
per schedule period and channel symbol error rate for flow 1 is varied.
In lower figure, channel symbol error rate for both flows is  $10^{-2}$,
and PHY rate for flow 1 is varied.}\label{fig:onecell2}
\end{figure}

\subsubsection{Mix of delay--sensitive and delay--insensitive flows}
Suppose the flows in the network belong to two classes, one of which is
delay--sensitive and has a delay--deadline $D$ whereas the other is
delay--insensitive i.e. has an infinite delay deadline. These classes
might correspond, for example, to video and data traffic. 
Figure~\ref{fig:onecell1}(a) plots the optimal airtime allocation as the
delay deadline $D$ is varied. In this example, there is a single
delay--sensitive flow and two delay--insensitive flows, and the airtime
allocation is shown for the delay--sensitive flow and for one of the
delay--insensitive flows (both receive the same airtime allocation). As
expected, it can be seen that the airtime allocations of the
delay--sensitive and delay--insensitive flows approach each other as the
delay deadline $D$ is increased. However, it is notable that they approach each other
fairly slowly, and when the delay deadline is low the airtime allocated
to the delay--sensitive flow is almost 50\% greater than that allocated
to a delay--insensitive flow.   This behaviour is qualitatively different
from the classical proportional fair allocation neglecting coding rate
and delay--deadlines, which would allocate equal airtimes to all flows.
By taking coding rate and delay deadlines into account, our approach
allows the resource allocation to flows with different quality of
service requirements  to be carried out in a principled and fair manner.

Figure \ref{fig:onecell1}(b) plots the optimal airtime allocation as the
number $N$ of delay--insensitive flows is varied.   It can be seen that
the airtime allocated to each flow decreases as $N$ is increased, as
expected since the number of flows sharing the network is increasing.
Interestingly, observe that the airtime allocated to the delay--sensitive
flow is a roughly constant margin above that allocated to the
delay--insensitive flows.   The delay--sensitive flow is therefore
``protected'' from the delay--insensitive flows.   However, in contrast
to ad hoc approaches, this protection is carried out in a
principled and fair manner.

\subsubsection{Mix of near and far stations}
Consider now a situation where all flows have the same delay deadline
$D$, but where for some flows the sources are located close to the
destination and for other flows the sources are further away. We
therefore have two classes of flows, one with a higher channel symbol
error rate than the other when both use the same PHY rate. Figure
\ref{fig:onecell2}(a)  plots the optimal airtime allocation for a flow
in each class as the channel error rate for one class is varied. When
the channel error rates for both classes is the same 
($\beta_f=10^{-2}$), it can be seen that the airtime allocation is the same. As the
channel error rate decreases, the airtime allocated to flow 1
decreases. Conversely, as the channel error rate increases, the
airtime allocated to flow 1 increases. 

Figure \ref{fig:onecell2}(b) plots the optimal airtime allocation when
flows in both classes have the same channel error rate but different PHY
rates i.e. where the PHY modulation has been adjusted to equalise the
channel error rates. When the PHY rates are the same ($w_{f,c}=10$
symbols per schedule period), the airtime allocation is the same to both
classes. As the PHY rate is increased, the airtime allocation for flow 1
decreases. Conversely, as the PHY rate is decreased, the airtime
allocation for flow 1 increases. Again, note that this is qualitatively different
from the classical proportional fair allocation neglecting coding rate
and delay--deadlines which would allocate equal airtimes to all flows.

\subsubsection{Unequal Airtimes}
\label{sec:unequal_airtimes}
The basic observations in these examples apply more generally.  In
particular, as noted above, in a loss-free, delay--insensitive single-cell network the proportional fair allocation is to assign equal air--time to all flows
(\cite{eq_air_time} and Section \ref{sec:lossfree}).   However, when delay deadlines are introduced and/or links are lossy, we see an interesting phenomenon. 
\begin{lemma}
\label{lem:unequal_air_time}
The optimum rate allocation ${\bm x}^*$ (or equivalently ${\bm r}^*$) is
not equivalent to an equal air--time allocation. 
\end{lemma}
\begin{proof}
See Appendix~\ref{app:unequal_air_time}
\end{proof}
In particular, flows that see a better channel get less
air--times than flows that see a worse channel. 

\subsection{Multiple cells}
\begin{figure}
\centering
\includegraphics[width=1.0\columnwidth]{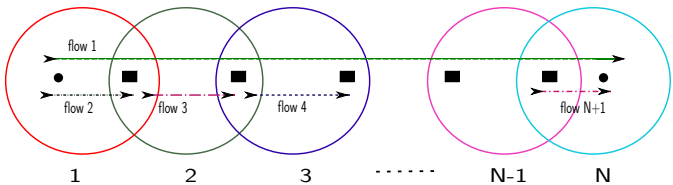}
   \caption{A linear Parking Lot network with $N$ cells and $N+1$ flows (one multi-hop flow and $N$ single-hop flows).}\label{multiple_cell}
\end{figure}

We now consider a mesh network consisting of $N$ cells carrying $N+1$
flows in the well-studied Parking Lot topology. The network topology is illustrated in
Fig.~\ref{multiple_cell}. The flows in this network can be
assigned to two classes: class 1 consists of the $N$-hop flow, and
class 2 consists of the single--hop flows 2, 3, $\cdots$, $N+1$.   Each cell has the same schedule period, i.e.
$T_c=T, \forall c\in\C$.

\subsubsection{Impact of number of hops}
Suppose that both classes of flow use the same symbol transmission PHY rate and experience the same loss rate in each cell.   Then the $N$-hop flow will experience a higher end-to-end symbol error rate that the single hop flows, and the loss rate will increase with $N$.   Fig.~\ref{MulticellE1} plots the ratio of
optimal airtimes allocated to each class of flow versus $N$.  Results are shown for three
delay deadline requirements: both classes of flow are delay--sensitive with delay deadline $D_1=D_2=1$; class
1 is delay--sensitive ($D_1=1$) while class 2 is delay--insensitive
($D_2=10^5$); class 1 is delay--insensitive while class 2 is
delay--sensitive. It can be seen that in the first case, where both classes have the
same delay deadline, the ratio of airtimes is larger than 1. This is in accordance with the
previous observation that flows with poorer channel
conditions are allocated more airtime than flows with better channel
conditions.  In the second case, where class 2 is delay--insensitive ($D_2=10^5$),
additional airtime is allocated to class 1, the delay--sensitive
flow, which also corresponds with the single cell analysis.
In the third case, where class 1 is delay--insensitive ($D_1=10^5$) and
class 2 is delay--sensitive, it can be seen that class 2 flows are
allocated slightly more airtime that the class 1 flow.   Interestingly,
however, observe that the airtime allocated to the class 1 flow is
insensitive to the number $N$ of hops.  This  contrasts with the
behaviour when the class 1 flow is delay--sensitive.

\subsubsection{Impact of different flow PHY rates}
Now consider a situation where the number of cells $N=3$ and all flows
have the same delay deadline $D_1=D_2=1$.   Flow 2 and flow 4 have
symbol error rate $10^{-4}$, and flow 1 and flow 3 have symbol error
rate $2.5\times10^{-1}$.   We classify the flows into three sets: class 1 consists of multi-hop flow 1, class 2 consists of single-hop flows 2 and 4, class 3 consists of single-hop flow 3.   Let $w$ denote the PHY rate used used by class 1 and class 2 flows, and $w_3$ denote the PHY rate used by the class 3 flow.   Fig.~\ref{MulticellE2} plots the optimal
coded packet size versus the ratio $w_3/w$.  We begin by observing that when $w_3/w=1$, all flows have the same PHY rate and  it can be seen that flows in classes 2 and 3 are allocated the same packet sizes (and so the same airtime).   Hence, although the flow in class 3 crosses a much more lossy link than the flows in class 2, the optimal allocation ensures that all of the single-hop flows have the same airtime.    The multi-hop flow in class 1 is allocated a smaller packet size (and so less airtime) than the single hop flows.  It can also be
seen that varying the PHY rate for the single-hop flow in class 3 does not affect the
optimal coded packet sizes of flows in class 1 and class 2, and hence the airtime
of class 1 and class 2 flows remains the same as $w_3$ is varied.    The
coded packet size of the class 3 flow increases linearly with $w_3/w$,
and so the airtime of the class 4 flow remains invariant as well.  

\begin{figure}
\centering
\includegraphics[width=0.8\columnwidth]{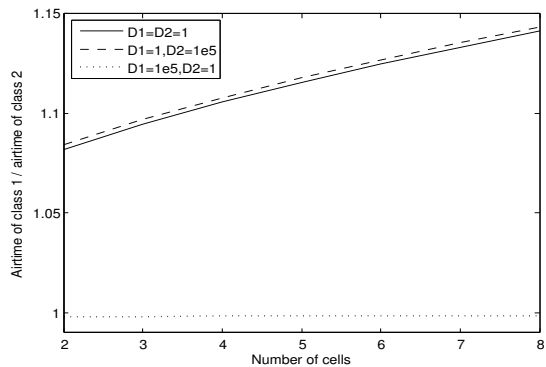}
   \caption{Ratio of airtimes vs. number $N$ of cells in Parking Lot topology of Fig.~\ref{multiple_cell}.  The y-axis is the ratio of the airtime allocated to the $N$-hop flow to that allocated to a single hop flow; note that the airtime of the $N$-hop flow is the sum of allocated airtime in each cell along the flows route.  Data is shown for three different delay deadline requirements, as indicated in the legend. All flows have the same PHY rate. }\label{MulticellE1}
\end{figure}

\begin{figure}
\centering
\includegraphics[width=0.8\columnwidth]{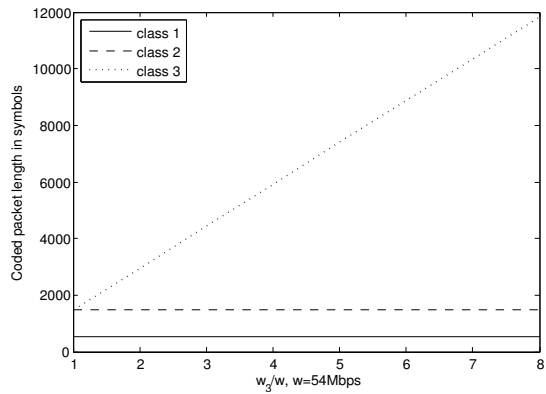}
   \caption{Coded packet size vs. ratio of PHY rates $w_3/w$ for Parking Lot topology of Fig.~\ref{multiple_cell} with $N=3$ cells.   Class 1 consists of multi-hop flow 1, class 2 consists of single-hop flows 2 and 4, class 3 consists of single-hop flow 3; class 1 and 2 flows use PHY rate $w$ bit/sec, the class 3 flow uses a PHY rate of $w_3$ bit/sec.   All flows have delay deadline $D=1$.}\label{MulticellE2}
\end{figure}

\section{Conclusions}
\label{sec:conclusions}
In this paper, we posed a utility fair problem that yields the optimum
airtime and the coding rate across flows in a capacity constrained multi-hop
network with delay deadlines.  We showed that the problem is highly non--convex.  Nevertheless, we
demonstrate that the global network utility
optimisation problem can be solved.  We obtained the optimum airtime/packet
size, channel coding rate, and analysed its properties. We also
analysed some simple networks based on the utility optimum framework we
proposed. To the best of our knowledge, this is the first work on
cross--layer optimisation that studies optimum coding across flows which
are competing for network resources and have delay--deadline constraints.

\appendices
\section{Proof of Lemma~\ref{lem:pe_bound}}
\label{app:proof_lemma_bound_pe}

\begin{align*}
\widetilde{e}_f 
&= \myprob{\sum_{i=1}^{D_fn_f}E_f[i] > \frac{D_fn_f-D_fk_f}{2}}\\
&= \myprob{\sum_{i=1}^{D_fn_f}E_f[i] > D_fn_f\left(\frac{1-r_f}{2}\right)}\\
&\leqslant \exp\left(-D_fn_f\cdot\theta_f\frac{1-r_f}{2}\right)
\myexp{\exp\left(\theta_f\sum_{i=1}^{D_fn_f}E_f[i]\right)} \\
&= \exp\left(-D_fn_f\cdot\theta_f\frac{1-r_f}{2}\right) \left[
\myexp{\exp\bigg(\theta_f E_f[1]\bigg)}\right]^{D_fn_f} \\
&= \exp\left(-D_fn_f\cdot\theta_f\frac{1-r_f}{2}\right)
\exp\left( D_fn_f\ln\bigg(\myexp{e^{\theta_f E_f[1]}}\bigg)\right) \\
&= \exp\left(-D_fn_f\left[\theta_f\frac{1-r_f}{2} -  \ln\bigg(\myexp{e^{\theta_f E_f[1]}}\bigg)\right]\right) \\
&= \exp\left(-D_fn_f
I_{E_f[1]}\left(\frac{1-r_f}{2};\theta_f\right)\right)
\end{align*}
$\hfill\blacksquare$

\section{Proof of Lemma~\ref{lem:pe_lower_bound}}
\label{app:proof_lemma_lower_bound_pe}.
\begin{align*}
\widetilde{e}_f 
&=    \ \myprob{\sum_{i=1}^{D_fn_f}E_f[i] > \frac{D_fn_f-D_fk_f}{2}} \nn
&=    \ \myprob{\sum_{i=1}^{D_fn_f}E_f[i] > D_fn_f \frac{1-r_f}{2}} \nn
&= \sum_{i=D_fn_f\frac{1-r_f}{2}+1}^{D_fn_f}
{D_fn_f \choose i} \beta_f^i (1-\beta_f)^{D_fn_f - i}. 
\end{align*}
The binomial coefficients can be bounded as follows:
\begin{align*}
1 \leqslant  \ {n \choose k} \ = \ \frac{n(n-1)\cdots(n-k+1)}{1\cdot2\cdots k} 
\ \leqslant n^k. 
\end{align*}
Hence,
\begin{align*}
 \widetilde{e}_f 
&\geqslant \sum_{i=D_fn_f\frac{1-r_f}{2}+1}^{D_fn_f}
\beta_f^i (1-\beta_f)^{D_fn_f- i} \\
& \geqslant \frac{\beta_f}{1-\beta_f} \beta_f^{D_fn_f(\frac{1-r_f}{2})}
(1-\beta_f)^{D_fn_f(\frac{1+r_f}{2})} \\
& = \frac{\beta_f}{1-\beta_f} 
\exp\left(-D_fn_f\left[ \left(\frac{1-r_f}{2}\right) \ln(1/\beta_f)
\right]\right)\nn  
& \exp\left(-D_fn_f\left[ 
 \left(\frac{1+r_f}{2}\right) \ln(1/(1-\beta_f))  \right]\right)\\
& = \frac{\beta_f}{1-\beta_f} \exp\left(-D_fn_f\left[
 \left(\frac{1-r_f}{2}\right)
\ln\left(\frac{1-r_f}{2\beta_f}\right) \right]\right)\nn
& \exp\left(-D_fn_f\left[
 \left(\frac{1+r_f}{2}\right)
\ln\left(\frac{1+r_f}{2(1-\beta_f)}\right)  \right]\right)\nn
& \ \ \ \
\cdot \exp\left(D_fn_f\left[
 \left(\frac{1-r_f}{2}\right)
\ln\left(\frac{1-r_f}{2}\right)\right]\right)\nn
& 
\cdot \exp\left(D_fn_f\left[
 \left(\frac{1+r_f}{2}\right)
\ln\left(\frac{1+r_f}{2}\right)  \right]\right)\\
& = \frac{\beta_f}{1-\beta_f} 
 \exp\left(-D_fn_f H\left({\cal
 B}\left(\frac{1-r_f}{2}\right)\right)\right)\nn
 &
\exp\left(-D_fn_f I\left({\cal B}\left(\frac{1-r_f}{2}\right)\|{\cal
B}(\beta_f)\right)\right)
\end{align*}
\normalsize
$\hfill\blacksquare$


\section{Proof of Lemma~\ref{lem:opt_joint_vs_sep}}
\label{app:proof_lemma_opt_joint_vs_sep}
For any $z \in {\cal Z}$, the function $f(y,z)$ is concave in $y$.
Hence, for each $z$, there exists a unique maximum $y^+(z)$, which   
is given by
\begin{eqnarray*}
f(y^+(z),z) & = & \max_{y\in{\cal Y}} \ f(y,z)\nn
& =: & g(z)
\end{eqnarray*}
If $f(y^+(z),z)$ is a concave function of $z$, then there exists a
unique maximiser, which is denoted by $z^+$, i.e., 
\begin{eqnarray*}
z^+ & = & \arg \max_{z \in {\cal Z}} \ f(y^+(z),z).
\end{eqnarray*}
We show that $(y^+(z^+),z^+)$ is an optimum solution to 
Eqn.~\eqref{eqn:joint_opt_prob}. Since $z^+$ is the maximiser of $g$, we
have for any $z \in {\cal Z}$,
\begin{eqnarray*}
g(z^+) & \geqslant & g(z) \\ 
\text{or} \ \ \  f(y^+(z^+),z^+) & \geqslant & f(y^+(z),z). 
\end{eqnarray*}
For any given $z \in {\cal Z}$, $y^+(z)$ is the maximiser of
$f(y,z)$ over all $y \in {\cal Y}$, i.e.,
\begin{eqnarray*}
f(y^+(z),z) & \geqslant & f(y,z),
\end{eqnarray*}
and hence, for all $(y,z) \in {\cal Y}\times{\cal Z}$,
\begin{eqnarray*}
f(y^+(z^+),z^+) \ \geqslant \ f(y^+(z),z) \ \geqslant \ f(y,z).
\end{eqnarray*}
We note that $y^+(\cdot)$ maps ${\cal Z}$ into
${\cal Y}$, and hence, $(y^+(z^+),z^+) \in {\cal Y}\times{\cal Z}$.
Hence, $(y^+(z^+),z^+)$ is a global maximiser.
$\hfill\blacksquare$

\section{Proof of Lemma~\ref{lemma:convexity_of_log_tx}}
\label{app:proof_lemma_convexity_of_log_tx}
Consider the optimisation problem, 
\begin{eqnarray*}
\max_{\widetilde{\bm n},\widetilde{\bm I}} & & \sum_{f \in \F}
\widetilde{n}_f + \ln(1-2x_f(\widetilde{I}_f)) + \ln(1-e^{-D_f
\exp(\widetilde{n}_f+\widetilde{I}_f)})\nn
\text{s.t.} \  & &  \sum_{f: c \in \C_f}
\frac{e^{\widetilde{n}_f}}{w_{f,c}} \leqslant T_c, \ \ \forall c \in \C
\end{eqnarray*}

We show that the objective function is jointly (strictly) concave in 
$(\widetilde{\bm n},\widetilde{\bm I})$. The objective
function is separable in 
$(\widetilde{n}_f,\widetilde{I}_f)$, and we show that 
$x_f(\widetilde{I}_f)$ is convex, and 
$\ln(1-e^{-D_f\exp(\widetilde{n}_f+\widetilde{I}_f)})$ is concave.

Since, for $x_f \in (\beta_f, 0.5)$, $I_f$ is a monotone function of
$x_f$, and $\widetilde{I}_f$ is a monotone function of
$I_f$, it is clear that $\widetilde{I}_f$ is invertible. Note
that 
\begin{eqnarray*}
\widetilde{I}_f & = & \ln(I_f) \nn
\frac{d\widetilde{I}_f}{dx_f} & = & \frac{\theta^*_f(x_f)}{I_f} \nn
\frac{dx_f}{d\widetilde{I}_f} & = & \frac{I_f}{\theta^*_f(x_f)} \nn
\frac{d^2x_f}{d\widetilde{I}_f^2} & = & \frac{I_f}{\theta^*_f(x_f)} 
 \left[
1 - 
 \frac{I_f}{\theta^*_f(x_f)^2} \frac{1}{x_f(1-x_f)}
\right]
\end{eqnarray*}
Define $g(x_f) := x_f(1-x_f)\theta_f^*(x_f)^2-I_f$. If $g(x_f) > 0$,
then $x_f(\widetilde{I}_f)$ is (strictly) convex.  
Note that 
$g'(x_f) = (1-2x_f)\theta_f^*(x_f)^2 + \theta_f^*(x_f)$ 
is increasing with $x_f$, and hence, $g(x_f) > g(\beta_f)=0$.

Define $h(x,y)= e^{x+y}$. consider the function 
\begin{eqnarray*}
f(x,y) & = & \ln(1-e^{-D_f g(x,y)}) \nn
\frac{\partial f}{\partial x} & = & \frac{D_fe^{-D_f g}}{1-e^{-D_f g}}
\frac{\partial g}{\partial x}\nn
\frac{\partial f}{\partial x} & = & \frac{D_f g e^{-D_f g}}{1-e^{-D_f
g}}.\nn
\text{Similarly,} \ \ \   
\frac{\partial f}{\partial y} & = & \frac{D_f g e^{-D_f g}}{1-e^{-D_f
g}}.
\end{eqnarray*}
Also,
\begin{eqnarray*}
\frac{\partial^2 f}{\partial x \partial y} & = & 
\frac{-D_f^2 g^2 e^{-2D_f g}}{(1-e^{-D_f g})^2} + 
\frac{D_f g e^{-D_f g}}{1-e^{-D_f g}} - 
\frac{-D_f^2 g^2 e^{-D_f g}}{1-e^{-D_f g}}\nn
&=&
\frac{-D_f g}{(1-e^{-D_f g})^2}\cdot \left[
D_f g e^{-2D_fg} + \right. \nn
& & \left. D_f g (1-e^{-D_fg})e^{-D_fg} - (1-e^{-D_fg})e^{-D_fg}
\right]\nn 
&=&
\frac{-D_f g}{(1-e^{-D_f g})^2} \left[
D_f g  - (1-e^{-D_fg})e^{-D_fg}
\right]
\end{eqnarray*}
Similarly, one can show that 
$\frac{\partial^2 f}{\partial x^2}  =  
\frac{\partial^2 f}{\partial y^2}  =  
\frac{\partial^2 f}{\partial x\partial y}$. Define $\ell(x) =
D_f x  - (1-e^{-D_fx})e^{-D_fx}$. If $\ell(x) > 0$, then $f(x,y)$ is
(strictly) convex. Note that $\ell'(x) = (D_f - D_fe^{-2D_fx}) + (D_fe^{-D_f x} -
D_f e^{-2D_fx}) > 0$. Therefore, $\ell(x) > \ell(0) = 0$.
$\hfill\blacksquare$

\section{Proof of Lemma~\ref{lem:unequal_air_time}}
\label{app:unequal_air_time}

From Eqn.~\eqref{eqn:kkt-n}, it is clear that even for a single cell,
because of the non--zero second term in the LHS, the air--time of flow
$f$ given by $\frac{n_f}{w_{f,c}}$ is not the same for all the flows
$f$. 
$\hfill\blacksquare$
\bibliographystyle{IEEEtran}
\bibliography{IEEEabrv,allerton}

\end{document}